%% file: Exact_discrete_lagrangianin_contact_mechanics.tex
\documentclass[11pt]{article}  

\usepackage{graphicx}

\usepackage{amssymb}
\usepackage{latexsym}
\usepackage[english]{babel}
\usepackage{mathtools}
\usepackage{makeidx}
\usepackage[utf8]{inputenc}
\usepackage{verbatim}
\usepackage{cite}
\usepackage[pagewise]{lineno}
\usepackage{amsthm}
\usepackage{pgfplots}
\pgfplotsset{compat=1.14}
\usepackage{tikz-cd} 
\usepackage[hyperindex,breaklinks]{hyperref}

\newtheorem{theorem}{Theorem}[section]

\newtheorem{lemma}[theorem]{Lemma}
\newtheorem{proposition}[theorem]{Proposition}

\theoremstyle{definition}
\newtheorem{definition}[theorem]{Definition}
\newtheorem{remark}[theorem]{Remark}
\newtheorem{example}{Example}

\newcommand{\R}{\ensuremath{\mathbb{R}}}
\newcommand{\N}{\ensuremath{\mathbb{N}}}

\newcommand{\C}{\ensuremath{\mathcal{C}}}

\newcommand{\M}{\ensuremath{\mathcal{M}}}

\newcommand{\F}{\ensuremath{\mathbb{F}}}

\DeclarePairedDelimiter{\set}{\{}{\}}

\newcommand{\Reeb}{\mathcal{R}}
\newcommand{\dd}{\mathrm{d}}
\newcommand{\dv}[2]{\frac{\dd #1}{\dd #2}}
\newcommand{\pdv}[2]{\frac{\partial #1}{\partial #2}}
\newcommand*{\contr}[1]{i_{#1}}

\newcommand*{\liedv}[1]{\mathcal{L}_{#1}}

\begin{document}
	\title{On the geometry of discrete contact mechanics}
	\author{
		{\bf\large Alexandre Anahory Simoes}\hspace{2mm}
		\vspace{1mm}\\
		{\it\small Instituto de Ciencias Matematicas }\\
		{\it\small Calle Nicolas Cabrera, 13-15, Campus Cantoblanco, UAM}, {\it\small 28049 Madrid, Spain}\\
		\vspace{2mm}\\
		{\bf\large David Martín de Diego}\hspace{2mm}
		\vspace{1mm}\\
		{\it\small Instituto de Ciencias Matematicas (CSIC-UAM-UC3M-UCM) }\\
		{\it\small Calle Nicolas Cabrera, 13-15, Campus Cantoblanco, UAM}, {\it\small 28049 Madrid, Spain}\\
		\vspace{2mm}\\
		{\bf\large Manuel Lainz Valcázar}\hspace{2mm}
		\vspace{1mm}\\
		{\it\small Instituto de Ciencias Matematicas }\\
		{\it\small Calle Nicolas Cabrera, 13-15, Campus Cantoblanco, UAM}, {\it\small 28049 Madrid, Spain}\\
		\vspace{2mm}\\
		{\bf\large Manuel de León}\hspace{2mm}
		\vspace{1mm}\\
		{\it\small Instituto de Ciencias Matematicas and Real Academia Española de Ciencias}\\
		{\it\small Calle Nicolas Cabrera, 13-15, Campus Cantoblanco, UAM}, {\it\small 28049 Madrid, Spain}\\
	}

\date{\today}

\maketitle

\begin{abstract}
In this paper, we continue the construction of variational integrators adapted to contact geometry started in \cite{VBS}, in particular, we introduce a discrete Herglotz Principle and the corresponding discrete Herglotz Equations for a discrete Lagrangian in the contact setting. This allows us to develop convenient numerical integrators for contact Lagrangian systems that are conformally contact by construction. The existence of an exact Lagrangian function is also discussed. 
\end{abstract}

\let\thefootnote\relax\footnote{\noindent AMS {\it Mathematics Subject Classification (2010)}. Primary 34C25; Secondary  37E40,
	70K40.
	}

\section{Introduction}

Contact Hamiltonian and Lagrangian systems have deserved a lot of attention in recent years \cite{Bravetti2017,Bravetti2018} or \cite{deLeon2019}.
One of the most relevant features of contact dynamics is the absence of conservative properties
contrarily to the conservative character of the energy in symplectic dynamics; indeed, 
we have a dissipative behavior. This fact suggests that contact geometry may be the appropriate framework to model many physical and mathematical problems with dissipation we find in thermodynamics, statistical physics, quantum mechanics, gravity or control theory, among many others. Consequently, it becomes an important necessity to develop numerical methods adapted to the contact setting for applications in the above mentioned subjects.
The idea is to develop geometric integrators, that is,  numerical methods for differential equations which preserve geometric properties like contact structure, symmetries, configuration space...  This preservation of structural properties is often desirable to achieve correct qualitative behavior and long time stability \cite{hairer,serna,blanes}.

As far as we know, the first attempt to develop geometric integrators for the contact case is in the paper \cite{VBS} (see also \cite{BSVZ}), 
where the authors present geometric numerical integrators for contact flows that stem from a discretization of Herglotz variational principle.

Our goal in the current paper is to go further in the discrete description of contact dynamics,
so we will mention some of the new and relevant results that the reader can find in the next pages.
Instead of deriving the discrete Herglotz equations by an heuristic argument, they are directly
obtained from a clear discrete variational principle. In addition, to develop the discrete algorithm
we use the natural discretization $Q \times Q \times \mathbb{R}$, which preserves all the contact geometry flavor.

Another relevant point is the discussion of the existence of an exact discrete Lagrangian function \cite{marsden-west,PatrickCuell}, which will lead us to define the contact exponential map and prove its existence. This construction is essential to develop a complete theory of variational error analysis for contact Lagrangian systems. 

Finally, we consider a discrete version of the infinitesimal symmetries discussed in \cite{Gaset2019}, jointly
with the corresponding dissipated quantities.

The paper is structured as follows. Section 2 is devoted to a quick review of contact Hamiltonian and Lagrangian systems in the continuous setting.
In particular, we recall the Herglotz variational principle, since it will be the motivation to develop the corresponding discrete version. 
Section 3 is devoted to construct the discrete version of contact Lagrangian dynamics for a discrete Lagrangian $L_d : Q \times Q \times \mathbb{R} \to \mathbb{R}$, where $Q$ is the configuration manifold. We consider the discrete Herglotz principle to obtain the so-called discrete Herglotz equations.
The Legendre transformations $F^{-}L_d$ and  $F^{+}L_d$ are defined, and consequently the
discrete flow (at the Lagrangian and Hamiltonian levels); the main result is that the discrete flow is a conformal
contactomorphism. In Section 4 we define the contact exponential map for the Herglotz vector field and prove that it is a local diffeomorphism.
This result permits to study the existence of an exact Lagrangian function.
Finally, we consider several examples to illustrate our theoretical developments.

\section{Continuous contact mechanics}

\subsection{Contact manifolds and Hamiltonian systems}
In this section we will recall the main definitions and results on the theory of contact manifolds and Hamiltonian system. See~\cite{deLeon2018} for a more detailed overview.

A \emph{contact manifold} $(M,\eta)$ is an $(2n+1)$-dimensional manifold with a \emph{contact form} $\eta$ \cite{marle}. That is, $\eta$ is a $1$-form on $M$ such that $\eta \wedge \dd \eta^n$ is a volume form. This type of  manifolds have a distinguished vector field: the so-called Reeb vector field $\Reeb$, which is the unique vector field that satisfies:
\begin{equation}
	\contr{\Reeb} \dd \eta = 0, \quad \eta(\Reeb)=1.
\end{equation}

On a contact manifold $(M,\eta)$, we  define the following isomorphism of vector bundles:
\begin{equation}
	\begin{aligned}
		\flat: \quad&TM &\longrightarrow& T^* M,\\
	  	        & \quad v &\longmapsto& \contr{v}\dd \eta + \eta(v)\eta.
	\end{aligned}
\end{equation}
Notice that $\flat(\Reeb) = \eta$.

There is a Darboux theorem for contact manifolds. In a neighborhood of each point in $M$ one can find local coordinates $(q^i, p_i, z)$ such that
\begin{equation}
	 \eta = \dd z - p_i  \dd q^i.
\end{equation}
In these coordinates, we have
\begin{equation}
	\Reeb = \frac{\partial}{\partial z}.
\end{equation}

An example of a contact manifold is $T^*Q\times \R$. Here, the contact form is given by 
\begin{equation}\label{eq:cotangent_contact_structure}
	\eta_Q = \dd z - \theta_Q = \dd z - p_i \dd q^i,
\end{equation}
where $\theta_Q$ is pullback the tautological $1$-form of $T^*Q$, $(q^i,p_i)$ are natural coordinates on $T^*Q$ and $z$ is the $\R$-coordinate.

We say that a (local) diffeomorphism between two contact manifolds $F:(M,\eta) \to (N,\tau)$ is a (local) \emph{contactomorphism} if $F^* \tau = \eta$. We say that $F$ is a (local) \emph{conformal contactomorphism} if $F^* \ker \tau = \ker \eta$ or, equivalently, $F^* \tau = \sigma \eta$, where $\sigma: M \to \R \setminus \set{0}$ is the \emph{conformal factor}.

We say that a vector field $X$ on $M$ is an \emph{infinitesimal (conformal) contactomorphism} if its flow $F_t$ consists of (conformal) contactomorphisms.

From the general identify, where $F_t$ is a flow and $X$ is its infinitesimal generator
\begin{equation}
	\pdv{}{t} F_t^* \eta  =
	 F_t^* \liedv{X} \eta,
\end{equation}
we deduce that $X$ is infinitesimal contactomorphism if and only if
\begin{equation}
	\liedv{X} \eta = 0.
\end{equation}
Furthermore, $X$ is a conformal contactomorphism if and only if 
\begin{equation}
	\liedv{X} \eta = a \eta,
\end{equation}
for some $a:M\to\R$. The function $a$ is related to the conformal factors $\sigma_t$ of the conformal contactomorphisms $F_t$ by
\begin{equation}
	\sigma_t (x)= \int_0^t \exp(a(F_\tau(x))) \dd \tau.
\end{equation}

Given a smooth function $f:M\to\R$, its \emph{Hamiltonian vector field} $X_f$ is given by
\begin{equation}
	\flat(X_f) =  \dd f - (f + \Reeb(f)) \eta.
\end{equation}

A vector field $X$ is the Hamiltonian vector field of some function $f$ if and only if it is an infinitesimal conformal contactomorphism. In that case $X=X_f$ for $f = - \eta(X)$. Moreover, $\liedv{X} \eta  = - \Reeb(f) \eta$. Hence $X$ is an infinitesimal contactomorphism if and only if $X = X_f$ for some function $f$ such that $\Reeb(f) = 0$.

We call the triple $(M, \eta, H)$ a \emph{contact Hamiltonian system}, where $(M,\eta)$ is a contact manifold and $H:M \to \R$ is the \emph{Hamiltonian function}.

In contrast to their symplectic counterpart, contact Hamiltonian vector fields do not preserve the Hamiltonian. In fact
\begin{equation}
	X_H(H) = -\Reeb(H) H.
\end{equation}

\subsection{Contact Lagrangian systems}\label{contact:lagrangian}
Now we review the Lagrangian picture of contact systems. In~\cite{deLeon2019} we give a more comprehensive description which also covers the case of singular Lagrangians.

Let $Q$ be an $n$-dimensional \emph{configuration manifold} and consider the \emph{extended phase space} $TQ \times \R$ and a \emph{Lagrangian function}  
 $L:TQ\times \R \to \R$. In this paper, we will assume that the Lagrangian is regular, that is, the Hessian matrix with respect to the velocities $(W_{ij})$ is regular where
 \begin{equation}\label{eq:hessian}
     W_{ij} = \frac{\partial^2 L}{\partial \dot{q}^i \partial \dot{q}^j },
 \end{equation}
 and $(q^i, \dot{q}^i,z)$ are bundle coordinates for $TQ \times \R$. Equivalently, $L$ is regular if and only if the one-form
\begin{equation}
    \eta_L = \dd z - \theta_L
\end{equation}
is a contact form. Here,
\begin{align}
    \theta_L &= S^* (\dd L) = \pdv{L}{\dot q^i} \dd {q}^i,
\end{align}
 where $S$ is the canonical vertical endomorphism on $TQ$ extended to $TQ \times \R$, that is, in local $TQ \times \R$ bundle coordinates,
\begin{equation}\label{eq:canonical_endomorphism}
    S = \dd q^i \otimes \frac{\partial}{\partial \dot{q}^i}.
\end{equation}

The energy of the system is defined by
\begin{equation}
    E_L = \Delta(L) - L = \dot{q}^i \pdv{L}{\dot{q}^i} - L,
\end{equation}
where $\Delta$ is the Liouville vector field on $TQ$ extended to $TQ\times \R$ in the natural way.

The Reeb vector field of $\eta_L$, which we will denoted by $\Reeb_L$ is given by
\begin{equation}\label{eq:flat_iso}
    \Reeb_L = \frac{\partial}{\partial z} - 
    W^{ij} \frac{\partial^2 L}{\partial \dot{q}^i \partial z} \pdv{}{\dot{q}^j},
\end{equation}
where $(W^{ij})$ is the inverse of the Hessian matrix with respect to the velocities $W^{ij}$ (Equation~\eqref{eq:hessian}).

The Hamiltonian vector field of the energy $E_L$ will be denoted $\xi_L = X_{E_L}$, hence
\begin{equation}
    \flat_L(\xi_L) = \dd E_L - (\Reeb_L(E_L)+E_L)\eta_L,
\end{equation}
where $\flat_L(v) = \contr{v} \dd \eta_L + \eta_L (v) \eta_L$ is the isomorphism defined in Equation~\eqref{eq:flat_iso} for this particular contact structure.

$\xi_L$ is a second order differential equation (SODE) (that is, $S(\xi_L) = \Delta$) and its solutions are just the ones of the Herglotz equations (also called generalized Euler-Lagrange equations) for $L$ (see~\cite{deLeon2019}):
\begin{equation}\label{eq:herglotz}
    \dv{}{t} \left({\pdv{L}{\dot{q}^i}}\right) - \pdv{L}{q^i} = \pdv{L}{\dot{q}^i} \pdv{L}{z}.
\end{equation}

There exists a \emph{Legendre transformation} for contact Lagrangian systems. Given the vector bundle $TQ\times \R \to Q \times \R$, one can consider the fiber derivative $\F L$ of $L:TQ \times \R \to \R$, which has the following coordinate expression in natural coordinates:
\begin{equation}
	\begin{aligned}
		\F L:TQ \times \R &\to T^*Q \times \R\\
		(q^i,\dot{q}^i,z) &\mapsto (q^i, \pdv{L}{\dot{q}^i},z).
	\end{aligned}
\end{equation}
If we consider the contact structure $\eta_Q$~\eqref{eq:cotangent_contact_structure} on $T^*Q \times \R$, and $\eta_L$ on $TQ\times \R$ then $\F L$ is a local contactomorphism. 

In the case that $\F L$ is a global contactomorphism, then we say that $L$ is \emph{hyperregular}. In this situation, we can define a Hamiltonian $H:T^*Q\times \R \to \R$ such that $E_L = H \circ \F L$ and the Lagrangian and Hamiltonian dynamics are $\F L$-related, that is, $\F L_* \xi_L = X_H$.

\subsubsection{Herglotz variational principle}
Equations (\ref{eq:herglotz}) can be derived from a modified variational principle~\cite{Herglotz1930}. In contrast to the symplectic case, the action is not a definite integral. The contact action is the value at the endpoint of solution to a non-autonomous ODE.

In~\cite{deLeon2019} we defined the action on the space of curves with fixed endpoints. However, for our purposes here it is more convenient to define the action on the space of all curves and all initial conditions and then restrict it to the appropriate submanifold.

Let $\Omega$ be the (infinite dimensional) manifold of curves on $Q$, $c:[0,1]\to Q$. We denote by  $\Omega(q_0, q_1) \subseteq \Omega$, where $q_0,q_1 \in Q$, the submanifold whose elements are the smooth curves $c \in \Omega$ such that $c(0)=q_0$, $c(1)=q_1$. The tangent space of $\Omega$ at a curve $c$ is given by vector fields over $c$. In the case of $T_c \Omega(q_0,q_1)$, the vector fields over $c$ vanish a the endpoints. That is,
\begin{align}
        T_c \Omega &=  \set{
            \delta v \in \C^\infty([0,1] \to TQ) \mid 
            \tau_Q \circ \delta v = c}, \\
        T_c \Omega(q_0,q_1) &=  \set{
            \delta c \in  T_c \Omega  \mid \delta c(0)=0, \, \delta c(1)=0 
            }.
\end{align}

We define the operator 
\begin{equation}
    \mathcal{Z}:\Omega \times \R \to \C^\infty ([0,1] \to \R),
\end{equation}
 which assigns to each curve and initial condition $(c,z_0)$ the curve $\mathcal{Z}_{z_0}(c)$ that solves the following ODE:
\begin{equation}\label{contact_var_ode}
\begin{dcases}
    \dv{\mathcal{Z}_{z_0}(c)}{t} &= L(c, \dot{c}, \mathcal{Z}_{z_0}(c)),\\
    \mathcal{Z}_{z_0}(c)(0) &= z_0.
    \end{dcases}
\end{equation}

Now we define the \emph{contact action functional} as the map which assigns to each curve $c$ and initial condition $z_0$, the solution to the previous ODE evaluated at the endpoint:
\begin{equation}\label{eq:contact_action}
    \begin{aligned}
        \mathcal{A}: \Omega\times \R &\to \R,\\
        (c,z_0) &\mapsto \mathcal{Z}_{z_0}(c)(1).
    \end{aligned}
\end{equation}

When restricted to $\Omega(q_0,q_1)\times \set{z_0}$, the critical points of $\mathcal{A}$ are the solutions to Herglotz equation. More precisely,
\begin{theorem}[Herglotz variational principle]
    Let $L: TQ \times \R \to \R$ be a Lagrangian function and let $c\in \Omega(q_0, q_1)$ and $z_0 \in \R$. Then, $(c,\dot{c}, \mathcal{Z}_{z_0}(c))$ satisfies the Herglotz equations~\eqref{eq:herglotz} if and only if $c$ is a critical point of $\mathcal{A}_{z_0}\vert_{\Omega(q_0, q_1)}$.
\end{theorem}

\begin{proof}
    We will compute $T\mathcal{Z}$. Let $c \in \Omega(q_0, q_1)$ be a curve and consider some tangent vector $\delta c \in T_c \Omega$. We will first compute the partial derivative with respect to $c$ by fixing $z_0\in \R$, and then we will fix the curve and compute the partial derivative with respect to the initial condition $z_0$. In order to simplify the notation, let $\chi=(c,\dot{c}, \mathcal{Z}_{z_0}(c))$ and let $\psi = T_c \mathcal{Z}_{z_0}(\delta v)$.
    
    Consider a curve $c_\lambda \in \Omega$ (that is, a smoothly parametrized family of curves) such that
    \begin{equation*}
        \delta c={\dv{c_\lambda}{\lambda}}\Big{\vert}_{\lambda=0}
    \end{equation*} 
    Since $\mathcal{Z}_{z_0}(c_\lambda)(0)=z_0$ for all $\lambda$, then $\psi(0)=0$.
    
    We compute the derivative of $\psi$ by interchanging the order of the derivatives using the ODE defining $\mathcal{Z}$:
    \begin{align*}
        \dot{\psi}(t) &= 
        {\dv{}{\lambda}\dv{}{t} 
        \mathcal{Z}_{z_0}(c_\lambda(t))}\vert_{\lambda=0} \\&=
        \dv{}{\lambda}{L(c_\lambda(t), \dot{c}_\lambda(t),\mathcal{Z}(c_\lambda)(t))}\vert_{\lambda=0} \\&=
        \pdv{L}{q^i}(\chi(t)) {\delta c} ^i(t) +
        \pdv{L}{\dot{q}^i}(\chi(t)) {\delta \dot{c}}^i(t) + 
        \pdv{L}{z}(\chi(t)) \psi(t).
    \end{align*}

Hence, the function $\psi$ is the solution to the ODE above. Explicitly
\begin{equation}
     \psi(t) = \frac{1}{\sigma(t)} \int_0^t \sigma(\tau) \left({
        \pdv{L}{q^i}(\chi(\tau)) {\delta c} ^i(\tau) + \pdv{L}{\dot{q}^i}(\chi(\tau)) {\delta \dot{c}}^i(\tau)
        } \right) \dd \tau,
\end{equation}
where
\begin{equation}\label{eq:sigma}
    \sigma(t) = \exp \left({-\int_0^t \pdv{L}{z}(\chi(\tau)) \dd \tau}\right) > 0.
\end{equation}

Integrating by parts and using that $\psi(0)=0$, we get the following expression
\begin{equation*}
	\begin{split}
		\psi(t) & = \delta {c}^i(t) \pdv{L}{\dot{q}^i}(\chi(t)) \\
		 & + \frac{1}{\sigma(t)}
		\int_0^t  {\delta c}^i (t) \left(
		\sigma(t) \pdv{L}{q^i}(\chi(t)) - 
		\dv{}{t} \left(\sigma(t)\pdv{L}{\dot{q}^i}(\chi(t))\right) 
		\right) \dd t \\
		& = \delta {c}^i(t) \pdv{L}{\dot{q}^i}(\chi(t)) \\
		&+ \delta {c}^i(t) \pdv{L}{\dot{q}^i}(\chi(t)) \\
		& + \frac{1}{\sigma(t)} \int_0^t  {\delta c} ^i(\tau) \sigma(\tau) \left(
		\pdv{L}{q^i}(\chi(\tau)) - 
		\dv{}{\tau} \pdv{L}{\dot{q}^i}(\chi(\tau))+  \pdv{L}{\dot{q}^i}(\chi(\tau)) \pdv{L}{z}(\chi(\tau)) \right) \dd \tau.
	\end{split}
\end{equation*}

Now we compute the partial derivative with respect to the initial condition $z_0$. We interchange the order of the derivatives
\begin{equation}
    \dv{}{t} \pdv{\mathcal{Z}_{z_0}(c)}{z_0} = 
    \pdv{L}{z} (c,\dot{c}, \mathcal{Z}(c)) \pdv{\mathcal{Z}_{z_0}(c)}{z_0}
\end{equation}
\end{proof}

If we solve for $\pdv{\mathcal{Z}_{z_0}(c)}{z_0}$ the ODE above using that $\pdv{\mathcal{Z}_{z_0}(c)}{z_0}(0)=1$, we notice that
\begin{equation}
    \pdv{\mathcal{Z}_{z_0}(c)}{z_0}(t) = 
    \exp \left({\int_0^t \pdv{L}{z}(\chi(\tau)) \dd \tau}\right) = \frac{1}{\sigma(t)} ,
\end{equation}
where $\sigma$ is defined in~\eqref{eq:sigma}.

\subsubsection{Symmetries and dissipated quantities on contact Lagrangian systems}\label{ssec:lie_group_lagrangian}
As explained in~\cite{Gaset2019, deLeon2020}, given a symmetry on a contact system, one does not obtain a conserved quantity, but a quantity $f$ that dissipates at the same rate as the Lagrangian.

Given a contact Hamiltonian system $(M,\eta, H)$, we say that a quantity $f:M\to \R$ is \emph{dissipated} if 
\begin{equation}\label{eq:dissipated_infinitesimal}
	\liedv{X_H} f = -\Reeb(H) f,
\end{equation}
or, equivalently,
\begin{equation}
	{\phi_t}^*(f) = \sigma_t,
\end{equation}
where $\phi$ is the flow of $X_H$ and $\sigma_t$, its conformal factor.

Notice that the quotient of two dissipated quantities (if it is well defined) is a conserved quantity. 

We end this section by stating a Noether theorem in this setting, which provides a link between symmetries of the Lagrangian and conserved quantities.

Let $L:TQ \times \R \to \R$ be a regular Lagrangian. Let $G$ be a Lie group  acting on $Q$
\begin{equation}
    \Phi: G\times Q \to Q.
\end{equation}

We defined the lifted action as
\begin{equation}
    \tilde{\Phi}: G\times TQ\times \R \to TQ\times \R,
\end{equation}
given by $\tilde{\Phi}(g, v_q, z) = (T_q\Phi(v_q), z)$ where $v_q\in T_qQ$. We denote by $\xi_{TQ\times\R}$ to the vector field on $TQ\times \R$ which is the infinitesimal generator by the lifted action of an element $\xi$ of the Lie algebra  $\mathfrak{g}$ of $G$.

We define the momentum map $J_L$:
\begin{equation}
    \begin{aligned}
        J_L: TQ \times \R &\to \mathfrak{g}^*,\\
       \langle J_L(v_q,z), \xi\rangle & = - \eta_L(\xi_{TQ \times \R}).   
    \end{aligned}
\end{equation}
and we define $\hat{J}(\xi): TQ\times \R\rightarrow \R$ by  $\hat{J}(\xi)(v_q, z) =\langle J_L(v_q,z), \xi\rangle$.

Then we have the following~\cite[Section~4.1]{deLeon2020}
\begin{theorem}
	Let the lifted action $\tilde{\Phi}$ preserve the Lagrangian $L$, then $\tilde{\Phi}$ acts by contactomorphisms on $(TQ\times\R,\eta_L,E_L)$ and $\hat{J}(\xi)$	is a dissipated quantity for every $\xi \in \mathfrak{g}$.
\end{theorem}
 
\section{Discrete contact mechanics}

In this section, we will extend the approach to discrete mechanics as in\cite{marsden-west} to the case of contct dynamics (see also \cite{VBS}).

Let $L_{d}:Q\times Q \times \R\rightarrow \R$ be a discrete Lagrangian function. In our point of view $Q\times Q \times \R$ will be the discrete space corresponding to the manifold $TQ\times \R$, where continuous contact Lagrangian mechanics takes place. We fix a \textit{time-step} $h>0$, on which $L_{d}$ depends, though we will omit this explicit dependence.

For each $N \in\N$, let us define the \textit{discrete path space} as the space containing sequences on $Q$ with length $N+1$, i.e.,
\begin{equation*}
\C_{d}^{N}(Q)=\{ (q_{0},q_{1}, ..., q_{N} ) | q_{k}\in Q,\ k=0,\ldots, N \}.
\end{equation*}
The set $\C_{d}^{N}(Q)$ is a manifold and it is canonically identified with the product space $Q^{N+1}$.

To each $q_{d}\in \C_{d}^{N}(Q)$  and each $z_{0}\in\R$ we will associate another sequence $(z_{k})\in \R^{N+1}$ defined by
\begin{equation}\label{zkdefinition}
z_{k+1}-z_{k}=L_{d}(q_{k},q_{k+1},z_{k}), \quad k=0,..., N-1.
\end{equation}
In the sequel, for each $1\leqslant k \leqslant N$, we will denote by $\mathcal{Z}_{k}$ the function $\mathcal{Z}_{k}:Q \times Q \times \R\longrightarrow \R$
\begin{equation*}
         \mathcal{Z}_{k}(q_{k-1},q_{k},z_{k-1})=z_{k-1}+L_{d}(q_{k-1},q_{k},z_{k-1}).
\end{equation*}

We define the \textit{contact discrete action} to be the functional that for each point $q_{d}\in \C_{d}^{N}(Q)$ and each real number $z_{0}$ returns as output the real number $z_{N}$ obtained recursively from \eqref{zkdefinition}, i.e.,
\begin{equation}
	\begin{split}
		\mathcal{A}_{d}: \ \C_{d}^{N}(Q)\times \R & \longrightarrow \R \\
		(q_{d},z_{0}) & \mapsto z_{N}.
	\end{split}
\end{equation}

A \textit{variation} of a sequence $q_{d}\in \C_{d}^{N}(Q)$ is a curve $\widetilde q_{d}:(-\epsilon,\epsilon)\rightarrow \C_{d}^{N}(Q)$ satisfying $\widetilde q_{d}(0)=q_{d}$. Given such a variation, we will define its \textit{infinitesimal variation} by
\begin{equation*}
\delta q_{d} := \left. \frac{\dd}{\dd \epsilon}\right|_{\epsilon=0} \widetilde q_{d}(\epsilon)=(\delta q_{0}, ..., \delta q_{N}),
\end{equation*}
where $\delta q_{k}:=\left. \frac{\dd}{\dd \epsilon}\right|_{\epsilon=0} \widetilde{q}_{k} (\epsilon)$.

\begin{proposition}
Let $L_{d}$ be a smooth discrete Lagrangian. Then, if we fix $z_0\in \R$, we obtain the functional
\begin{equation*}
	\begin{split}
		\mathcal{A}_{d,z_{0}}: \ \C_{d}^{N}(Q) & \longrightarrow \R \\
		q_{d} & \mapsto \mathcal{A}_{d}(q_{d},z_{0}).
	\end{split}
\end{equation*}
The differential of the functional $\mathcal{A}_{d,z_{0}}$ is the following
	\begin{equation}\label{differential:discrete:action}
	    \begin{split}
	        \dd \mathcal{A}_{d,z_{0}}(q_{d})=& \sigma_{N}\cdot\cdot\cdot\sigma_{2}\frac{\partial \mathcal{Z}_{1}}{\partial q_{0}}(q_{0},q_{1},z_{0})\dd q_{0} \\
	        & +\sum_{k=1}^{N-1} \prod_{j=k+2}^{N}\sigma_{j}\cdot \left( \frac{\partial \mathcal{Z}_{k+1}}{\partial q_{k}}+\frac{\partial \mathcal{Z}_{k+1}}{\partial z_{k}}\frac{\partial \mathcal{Z}_{k}}{\partial q_{k}} \right) \dd q_{k} \\
	        & +\frac{\partial \mathcal{Z}_{N}}{\partial q_{N}}(q_{N-1},q_{N},z_{N-1})\dd q_{N},
	    \end{split}
	\end{equation}
where we are using the identification of $\C_{d}^{N}(Q)$ with $Q^{N+1}$ and for each $1 \leqslant j\leqslant N$
\begin{equation*}
    \sigma_{j}=\frac{\partial \mathcal{Z}_{j}}{\partial z_{j-1}}(q_{j-1},q_{j},z_{j-1}).
\end{equation*}
\end{proposition}


\begin{proof}
    Using the identification of $\C_{d}^{N}(Q)$ with $Q^{N+1}$, note that the discrete action may be rewritten as
    \begin{equation*}
        \mathcal{A}_{d,z_{0}}(q_{d})=\mathcal{Z}_{N}(q_{N-1},q_{N},\mathcal{Z}_{N-1}(q_{N-2},q_{N-1},\mathcal{Z}_{N-2}(...\mathcal{Z}_{1}(q_{0},q_{1},z_{0})...)).
    \end{equation*}
    Using that
    \begin{equation*}
        \dd \mathcal{A}_{d,z_{0}}(q_{d})=\frac{\partial \mathcal{A}_{d,z_{0}}}{\partial q_{0}} \dd q_{0}+\sum_{k=1}^{N-1}\frac{\partial \mathcal{A}_{d,z_{0}}}{\partial q_{k}} \dd q_{k}+\frac{\partial \mathcal{A}_{d,z_{0}}}{\partial q_{N}} \dd q_{N}.
    \end{equation*}
    and applying the chain rule, we deduce that
    \begin{equation*}
        \frac{\partial \mathcal{A}_{d,z_{0}}}{\partial q_{0}}=\frac{\partial \mathcal{Z}_{N}}{\partial z_{N-1}}\cdot\cdot\cdot\frac{\partial \mathcal{Z}_{2}}{\partial z_{1}}\frac{\partial \mathcal{Z}_{1}}{\partial q_{0}},
    \end{equation*}
    since the function $\mathcal{Z}_{1}$ is the only one that depends on $q_{0}$ among all the $N$ functions $\mathcal{Z}_{k}$.
    It is also clear that
    \begin{equation*}
        \frac{\partial \mathcal{A}_{d,z_{0}}}{\partial q_{N}}=\frac{\partial \mathcal{Z}_{N}}{\partial q_{N}},
    \end{equation*}
    since none of the functions $\mathcal{Z}_{k}$ depend on $q_{N}$ except the function $\mathcal{Z}_{N}$.
    Finally if $1\leqslant k \leqslant N-1$ we have that
    \begin{equation*}
        \frac{\partial \mathcal{A}_{d,z_{0}}}{\partial q_{k}}=\frac{\partial \mathcal{Z}_{N}}{\partial z_{N-1}}\cdot\cdot\cdot\frac{\partial \mathcal{Z}_{k+2}}{\partial z_{k+1}}\left(\frac{\partial \mathcal{Z}_{k+1}}{\partial q_{k}}+\frac{\partial \mathcal{Z}_{k+1}}{\partial z_{k}}\frac{\partial \mathcal{Z}_{k}}{\partial q_{k}}\right),
    \end{equation*}
    where we applied the chain rule and the fact that the functions $\mathcal{Z}_{k+1}$ and $\mathcal{Z}_{k}$ are the only ones that depend on $q_{k}$. Hence, we finished the proof.
\end{proof}

\begin{remark}
Let us see the special case $N=2$, where we can directly compute the differential of the action:

Let $L_{d}$ be a smooth discrete Lagrangian. In the case where $N=2$, the differential of the discrete action function satisfies:
	\begin{equation}
		\begin{split}
			&\dd \mathcal{A}_{d,z_{0}} =  \left( D_{1}L_{d} (q_{1},q_{2},z_{1})+(1+D_{z}L_{d}(q_{1},q_{2},z_{1}) D_{2}L_{d}(q_{0},q_{1},z_{0}) \right) \dd q_{1} \\
		& \quad\ + D_{2}L_{d}(q_{1},q_{2},z_{1}) \dd q_{2} + (1+D_{z}L_{d}(q_{1},q_{2},z_{1}))D_{1}L_{d}(q_{0},q_{1},z_{0}) \dd q_{0}.
		\end{split}
	\end{equation}
\end{remark}



\begin{definition}[Discrete Herglotz Principle]
Given $z_{0}\in \R$, a discrete path $q_{d}=(q_{0},..., q_{N})$ in $\C_{d}^{N}(Q)$ is said to satisfy the \textit{Discrete Herglotz Principle} if $q_{d}$ is a critical value of the discrete action functional $\mathcal{A}_{d,z_{0}}$ among all paths in $\C_{d}^{N}(Q)$ with fixed end points $q_{0},q_{N}$.
\end{definition}

We will now obtain as a sufficient and necessary condition for a path to satisfy the discrete Herglotz principle, a set of equations called \textit{Discrete Herglotz equations} \cite{VBS}.
\begin{theorem}
Let $L_{d}$ be a discrete Lagrangian function such that $1+D_{z}L_{d}$ is non-vanishing everywhere. Given $z_{0}\in \R$, a discrete path $q_{d}\in \C_{d}^{N}(Q)$ satisfies the discrete Herglotz principle if and only if it satisfies
\begin{equation}\label{DHE}
\begin{split}
		&\hspace{-0.5cm}D_{1}L_{d} (q_{k},q_{k+1},z_{k})+(1+D_{z}L_{d}(q_{k},q_{k+1},z_{k})) D_{2}L_{d}(q_{k-1},q_{k},z_{k-1})=0,\\
    &z_{k}-z_{k-1}=L_{d}(q_{k-1},q_{k},z_{k-1}), 
    \end{split}
\end{equation}
for $k=1,...,N-1$.
\end{theorem}

\begin{proof}
    Let $q_{d}(\epsilon)$ be a variation of $q_{d}\in\C_{d}^{N}(Q)$ with fixed end-points $q_{0}$ and $q_{N}$. Then $q_{d}$ is a critical value of the discrete action functional if and only if
    \begin{equation*}
        \left.\frac{\dd}{\dd \epsilon}\right|_{\epsilon=0}(\mathcal{A}_{d,z_{0}}(q_{d}(\epsilon)))=\dd \mathcal{A}_{d,z_{0}}(\delta q_{d})=0.
    \end{equation*}
    By \eqref{differential:discrete:action} the last expression is equivalent to
    \begin{equation*}
       \sum_{k=1}^{N-1} \prod_{j=k+2}^{N} \sigma_{j}\cdot \left( \frac{\partial \mathcal{Z}_{k+1}}{\partial q_{k}}+\frac{\partial \mathcal{Z}_{k+1}}{\partial z_{k}}\frac{\partial \mathcal{Z}_{k}}{\partial q_{k}} \right) \delta q_{k} = 0.
    \end{equation*}
    Since the infinitesimal variations $\delta q_{k}$, $1\leq k\leq N-1$, are arbitrary we deduce
    \begin{equation*}
       \prod_{j=k+2}^{N} \sigma_{j}\cdot \left( \frac{\partial \mathcal{Z}_{k+1}}{\partial q_{k}}+\frac{\partial \mathcal{Z}_{k+1}}{\partial z_{k}}\frac{\partial \mathcal{Z}_{k}}{\partial q_{k}} \right) = 0.
    \end{equation*}
    Note that,
    \begin{equation*}
        \sigma_{j}=\frac{\partial \mathcal{Z}_{j}}{\partial z_{j-1}}(q_{j-1},q_{j},z_{j-1})=1+D_{z}L_{d}(q_{j-1},q_{j},z_{j-1})
    \end{equation*}
    is non-vanishing by hypothesis and
    \begin{equation*}
       \frac{\partial \mathcal{Z}_{k+1}}{\partial q_{k}}+\frac{\partial \mathcal{Z}_{k+1}}{\partial z_{k}}\frac{\partial \mathcal{Z}_{k}}{\partial q_{k}}= D_{1}L_{d}(q_{k},q_{k+1},z_{k})+\sigma_{k+1}D_{2}L_{d}(q_{k-1},q_{k},z_{k-1}),
    \end{equation*}
    from where the result follows.
\end{proof}

\begin{remark}
    The discrete principle introduced in \cite{VBS} is just the condition
    \begin{equation*}
        \frac{\partial \mathcal{Z}_{k+1}}{\partial q_{k}}+\frac{\partial \mathcal{Z}_{k+1}}{\partial z_{k}}\frac{\partial \mathcal{Z}_{k}}{\partial q_{k}}=0,
    \end{equation*}
    afer rewriting it in our notation. For discrete Lagrangian functions where $1+D_{z}L_{d}$ is non-vanishing, the condition above is equivalent to the Herglotz discrete principle.
\end{remark}


\subsection{Discrete Lagrangian flows and discrete Legendre transforms}

Given a discrete contact Lagrangian $L_{d}$, if $1+D_{z}L_{d}(q_0,q_1,z_0)$ does not vanish, we can define two maps called \textit{discrete Legendre transforms}: $\F^{\pm}L_{d}:Q \times Q \times \R\rightarrow T^{*}Q \times \R$ 
\begin{equation}
    \begin{split}
        & \F^{+}L_{d}(q_0,q_1,z_0)=(q_{1},D_{2}L_{d}(q_0,q_1,z_0),z_{0}+L_d(q_0, q_1, z_0)) \\
        & \F^{-}L_{d}(q_0,q_1,z_0)=\left(q_{0},-\frac{D_{1}L_{d}(q_0,q_1,z_0)}{1+D_{z}L_{d}(q_0,q_1,z_0)},z_{0} \right).
    \end{split}
\end{equation}

\begin{lemma}
$\F^{+}L_{d}$ is a local diffeomorphism if and only if $\F^- L_d$ is a local diffeomorphism.
\end{lemma}
\begin{proof}
	It is a direct consequence of the implicit function theorem.  
\end{proof}

The Legendre transforms allow us to rewrite discrete Herglotz equations (\ref{DHE}) as a momentum matching equations as in \cite{marsden-west}. Indeed, provided $1+D_{z}L_{d}(q_0,q_1,z_0)$ is not zero, we may write
\begin{equation}
\F^{+}L_{d}(q_0,q_1,z_0)=\F^{-}L_{d}(q_1,q_2,z_1).
\end{equation}

We have the following theorem
\begin{theorem}\label{3.7}
Suppose $1+D_{z}L_{d}(q_0,q_1,z_0)$ does not vanish. Given a discrete Lagrangian $L_{d}:Q\times Q \times \R\rightarrow \R$, there is a well-defined discrete Lagrangian flow $\Phi_{d}:Q\times Q\times \R\rightarrow Q\times Q \times \R$ if and only if $\F^{-}L_{d}$
is a local diffeomorphism where 
\[
\Phi_{d}=(\F^{-}L_{d})^{-1}\circ  \F^{+}L_{d}
\]
\end{theorem}
\begin{proof}
	Given $(q_0, q_1, z_0)\in Q\times Q\times \R$, from the definition
	$\Phi_{d}=(\F^{-}L_{d})^{-1}\circ  \F^{+}L_{d}$ it is easy to show that 
 the point $(q_2, z_1)\in  Q\times \R$ required to satisfy
	\[
	 \Phi_d (q_0, q_1, z_0) =(q_1, q_2, z_1)
	 \]
is defined by Equations  (\ref{DHE}), and so $q_2$ and $z_1$  is uniquely defined as a
	function of $(q_0, q_1, z_0)$ if and only if $\F^{-}L_{d}$ is locally an isomorphism.
	\end{proof}

Inspired by the last theorem, we say that a discrete contact Lagrangian is \textit{regular} if $1+D_{z}L_{d}(q_0,q_1,z_0)$ does not vanish and its negative discrete Legendre transform $\F^{-}L_{d}$ is a local diffeomorphism.

The discrete Legendre transforms also allow us to define an associated \textit{discrete Hamiltonian flow} on $T^{*}Q\times \R$. 
Indeed, considering a regular discrete Lagrangian function $L_{d}$, let $\widetilde{\Phi_{d}}:T^{*}Q\times \R\rightarrow T^{*}Q\times \R$ be defined by
\begin{equation}
\widetilde{\Phi_{d}}=\F^{+}L_{d}\circ \Phi_{d} \circ (\F^{+}L_{d})^{-1}.
\end{equation}
It is not difficult to show that the discrete Hamiltonian flow admits the alternative expressions
\begin{equation}
\widetilde{\Phi_{d}}=\F^{-}L_{d}\circ \Phi_{d} \circ (\F^{-}L_{d})^{-1} 
\qquad \hbox{or}\qquad 
\widetilde{\Phi_{d}}=\F^{+}L_{d}\circ (\F^{-}L_{d})^{-1}.
\end{equation}

\begin{equation}
\begin{tikzcd}[column sep=1]
                                                            & Q \times Q \times \mathbb{R} \arrow[rr, "\Phi_{d}"] \arrow[ld, "\mathbb{F}^{-}L_{d}"] \arrow[rd, "\mathbb{F}^{+}L_{d}"] &                                                             & Q \times Q \times \mathbb{R} \arrow[ld, "\mathbb{F}^{-}L_{d}"] \arrow[rd, "\mathbb{F}^{+}L_{d}"] &                         \\
T^{*}Q\times \mathbb{R} \arrow[rr, "\widetilde{\Phi_{d}}"'] &                                                                                                                         & T^{*}Q\times \mathbb{R} \arrow[rr, "\widetilde{\Phi_{d}}"'] &                                                                                                  & T^{*}Q\times \mathbb{R}
\end{tikzcd}
\end{equation}

We may define the one-forms

\begin{equation}
	\eta^{+}=(\F^{+}L_{d})^{*}\eta, \quad \eta^{-}=(\F^{-}L_{d})^{*}\eta,
\end{equation}
where $\eta$ is the canonical contact form on $T^{*}Q\times \R$. These are contact forms on $Q\times Q \times \R$. If we chose natural coordinates $(q^{i},p_{i},z)$ on $T^{*}Q\times \R$ where $\eta=\dd z-p_{i}\dd q^{i}$, the discrete 1-forms may be locally written as the pullback
\begin{equation}
\begin{split}
\eta^{+}&=\dd z_{0}+\dd L_d(q_{0},q_{1},z_{0})-D_{2}L_{d}(q_{0},q_{1},z_{0})\dd q_{1}, \\ \eta^{-}&=\dd z_{0}+\frac{D_{1}L_{d}(q_0,q_1,z_0)}{1+D_{z}L_{d}(q_0,q_1,z_0)}\dd q_{0},
\end{split}
\end{equation}
by the corresponding discrete Legendre transform. The one-form $\eta^{+}$ is further simplified to

\begin{equation}\label{localetaplus}
\eta^{+}=(1+D_{z}L_{d}(q_{0},q_{1},z_{0}))\dd z_{0}+D_{1}L_{d}(q_{0},q_{1},z_{0})\dd q_{0}.
\end{equation}

Given a discrete Lagrangian $L_{d}$, let $\sigma_{d}:Q\times Q\times \R\rightarrow \R$ be the smooth function given by
\begin{equation*}
	\sigma_{d}(q_0,q_1,z_0)=1+D_{z}L_{d}(q_0,q_1,z_0)
\end{equation*}
then we have that:

\begin{lemma}\label{last-lemma}
The discrete contact forms $\eta^{\pm}$ satisfy
\begin{itemize}
\item[(i)] $\eta^{+}=\sigma_{d} \cdot \eta^{-}$;
\item[(ii)] $(\Phi_{d})^{*}\eta^{-}=\eta^{+}$.
\end{itemize}
\end{lemma}

\begin{proof}
	For the first item, observe that \eqref{localetaplus} is equivalent to
	\begin{equation*}
	\eta^{+}=(1+D_{z}L_{d}(q_{0},q_{1},z_{0}))\eta^{-}.
	\end{equation*}
	For the second one, note that
	\begin{equation*}
	(\Phi_{d})^{*}\eta^{-}=	(\Phi_{d})^{*}\circ (\F^{-}L_{d})^{*}\eta=(\F^{-}L_{d}\circ \Phi_{d})^{*}\eta=(\F^{+}L_{d})^{*}\eta
	\end{equation*}
	by applying Theorem \ref{3.7}.
\end{proof}

As a consequence of the last Lemma we have the following theorem:

\begin{theorem}
Let $L_{d}$ be a regular discrete Lagrangian function. The discrete flow $\Phi_{d}$ associated to $L_{d}$ is a conformal contactomorphism with respect to both contact structures $\eta^{\pm}$. In particular, it satisfies 
	\begin{equation}
		(\Phi_{d})^{*}\eta^{+}=(\sigma_{d} \circ \Phi_{d})\cdot \eta^{+}, \quad (\Phi_{d})^{*}\eta^{-}=\sigma_{d} \cdot \eta^{-}
	\end{equation}
Likewise, the discrete Hamiltonian flow $\widetilde{\Phi_{d}}$ is also a conformal contactomorphism satisfying
\begin{equation}
	(\widetilde{\Phi_{d}})^{*}\eta=(\sigma_{d} \circ (\F^{-}L_{d})^{-1}) \cdot \eta.
\end{equation}
\end{theorem}

\begin{proof}
	These are trivial consequences of Lemma \ref{last-lemma}. Combining the two statements of the Lemma we get
	\begin{equation*}
		(\Phi_{d})^{*}\eta^{-}=\sigma_{d} \cdot \eta^{-}.
	\end{equation*}
	Then, also
	\begin{equation*}
		(\Phi_{d})^{*}\eta^{+}=(\Phi_{d})^{*}(\sigma_{d} \cdot \eta^{-})=(\sigma_{d} \circ \Phi_{d}) \cdot (\Phi_{d})^{*}\eta^{-}=(\sigma_{d} \circ \Phi_{d}) \cdot \eta^{+}.
	\end{equation*}
	Observing that the discrete Hamiltonian flow satisfies $\widetilde{\Phi_{d}}=\F^{+}L_{d}\circ \Phi_{d} \circ (\F^{+}L_{d})^{-1}$ by definition, then
	\begin{equation*}
		\begin{split}
			(\widetilde{\Phi_{d}})^{*}\eta = & ((\F^{+}L_{d})^{-1})^{*} \circ (\Phi_{d})^{*} \eta^{+}=((\F^{+}L_{d})^{-1})^{*}((\sigma_{d} \circ \Phi_{d}) \cdot \eta^{+}) \\
			= & (\sigma_{d} \circ \Phi_{d}\circ (\F^{+}L_{d})^{-1}) \cdot ((\F^{+}L_{d})^{-1})^{*}\eta^{+},
		\end{split}
	\end{equation*}
	where the last equality comes from the properties of the pullback. Since we have that
	\begin{equation*}
		\Phi_{d}\circ (\F^{+}L_{d})^{-1}=(\F^{-}L_{d})^{-1} \quad \text{and} \quad ((\F^{+}L_{d})^{-1})^{*}\eta^{+}=\eta,
	\end{equation*}
	the desired result follows.
	
	Moreover, since the discrete Lagrangian function $L_{d}$ is regular, the function $\sigma_{d}$ does not vanish. Hence, the discrete flows $\Phi_{d}$ and $\widetilde{\Phi_{d}}$ are conformal contact.
\end{proof}

\subsection{Discrete symmetries and dissipated quantities}

Let $G$ be a Lie group acting on $Q$ through the map $\Phi:G\times Q\rightarrow Q$. We define the lifted action on $Q\times Q \times \R$ to be the diagonal action on $Q\times Q$ and the identity on $\R$, so that
\begin{equation*}
    \widetilde{\Phi}:G\times Q\times Q\times \R\rightarrow Q\times Q\times \R, \quad \widetilde{\Phi}_{g}(q_{0},q_{1},z_{0})=(\Phi_{g}(q_{0}),\Phi_{g}(q_{1}),z_{0}).
\end{equation*}
Let us denote by $\xi_{Q}\in\mathfrak{X}(Q)$ the infinitesimal generator associated to a Lie algebra element $\xi\in {\mathfrak g}$  and by $\widetilde{\xi}\in\mathfrak{X}(Q\times Q\times \R)$ the corresponding infinitesimal generator on $Q\times Q\times \R$.

Notice that, since $\text{pr}_{3}(\Phi_{g}(q_{0},q_{1},z_{0}))=z_{0}$ is constant for all $g\in G$, where $\text{pr}_{3}:Q\times Q\times \R\rightarrow \R$ is the projection onto the third factor, then we have that
\begin{equation*}
    T_{(q_{0},q_{1},z_{0})}\text{pr}_{3}(\widetilde{\xi}(q_{0},q_{1},z_{0}))=0.
\end{equation*}
In fact, the infinitesimal generator may be identified with
\begin{equation}\label{infinitesimalidentification}
    \widetilde{\xi}(q_{0},q_{1},z_{0})=(\xi_{Q}(q_{0}),\xi_{Q}(q_{1}),0_{z_{0}})\in T_{q_{0}}Q\times T_{q_{1}}Q\times T_{z_{0}}\R,
\end{equation}
where $0:\R\rightarrow T\R$ is the zero section of $T\R$.

\begin{lemma}
If $L_{d}:Q\times Q\times \R\rightarrow \R$ is an invariant discrete Lagrangian function, i.e., $L_{d}\circ\widetilde{\Phi}_{g}=L_{d}$ for all $g\in G$, then it satisfies the equation
    \begin{equation}\label{symmetryformula}
        D_{1}L_{d}(q_{0},q_{1},z_{0})\xi_{Q}(q_{0})+D_{2}L_{d}(q_{0},q_{1},z_{0})\xi_{Q}(q_{1})=0.
    \end{equation}
\end{lemma}

\begin{proof}
    Since the discrete Lagrangian function is invariant for the lifted action, it satisfies
    \begin{equation*}
        \langle \dd L_{d}(q_{0},q_{1},z_{0}), \widetilde{\xi}(q_{0},q_{1},z_{0})\rangle=0, \ \forall  (q_{0},q_{1},z_{0})\in Q\times Q\times \R.
    \end{equation*}
    Then using equation \eqref{infinitesimalidentification}, one immediately gets the desired expression.
\end{proof}

Now consider the discrete momentum map $J_{d}$ given by
\begin{equation}
    \begin{aligned}
        J_{d}: Q\times Q \times \R &\to \mathfrak{g}^*,\\
        \langle J_{d}(q_{0},q_{1},z_{0}), \xi\rangle & = \langle \eta^{-}, \widetilde{\xi}(q_{0},q_{1},z_{0})\rangle .   
    \end{aligned}
\end{equation}

\begin{theorem}
 Let $L_{d}$ be an invariant discrete Lagrangian function for the lifted action $\widetilde{\Phi}$. Then $\widetilde{\Phi}$ acts by contactomorphisms on $Q\times Q\times \R$ and the function $\hat{J}_{d}(\xi):Q\times Q\times \R\rightarrow \R$ given by
 \begin{equation*}
     \hat{J}_{d}(\xi)(q_{0},q_{1},z_{0})=\langle J_{d}(q_{0},q_{1},z_{0}), \xi\rangle
 \end{equation*}
 is dissipated along the discrete flow of Herglotz equations in the sense that
 \begin{equation*}
     \hat{J}_{d}(\xi)(\Phi_{d} (q_{0},q_{1},z_{0}))= \sigma_{d}(q_{0},q_{1},z_{0}) \hat{J}_{d}(\xi)(q_{0},q_{1},z_{0}),
 \end{equation*}
 where $\sigma_{d}(q_{0},q_{1},z_{0})=1+D_{z}L_{d}(q_{0},q_{1},z_{0})$.
\end{theorem}

\begin{proof}
    The fact that $\widetilde{\Phi}$ acts by contactomorphisms is immediately checked by computing the pullback of either the 1-forms $\eta^{\pm}$:
    \begin{equation*}
        (\widetilde{\Phi}_{g})^{*}\eta^{\pm}=
        \eta^{\pm}.
    \end{equation*}
as a direct consequence of the $G$-invariance of $L_d$ (see Subsection 1.3.3 in \cite{marsden-west}).

In order to simplify the notation, let $P_{0}=(q_{0},q_{1},z_{0})$ and $P_{1}=\Phi_{d} (q_{0},q_{1},z_{0})$. By definition we have that
    \begin{equation*}
        \hat{J}_{d}(\xi)(P_{1})=\langle \eta^{-}(P_{1}), \widetilde{\xi}(P_{1})\rangle.
    \end{equation*}
Now applying the definition of $\eta^{-}$ and equation \eqref{infinitesimalidentification} we get
    \begin{equation*}
        \hat{J}_{d}(\xi)(P_{1})=\frac{1}{\sigma_{d}(P_{1})}\langle D_{1}L_{d}(P_{1}), \xi_{Q}(q_{1})\rangle.
    \end{equation*}
Using the discrete Herglotz equations, the right-hand side reduces to
    \begin{equation*}
        \hat{J}_{d}(\xi)(P_{1})=-\langle D_{2}L_{d}(P_{0}),\xi_{Q}(q_{1})\rangle.
    \end{equation*}
From the infinitesimal symmetry formula in equation \eqref{symmetryformula}, we deduce
   \begin{equation*}
        \hat{J}_{d}(\xi)(P_{1})=\langle D_{1}L_{d}(P_{0}), \xi_{Q}(q_{0})\rangle .
    \end{equation*}
Now inserting $\sigma_{d}(P_{0})$ so that
    \begin{equation*}
        \hat{J}_{d}(\xi)(P_{1})=\sigma_{d}(P_{0}) \langle \frac{D_{1}L_{d}(P_{0})}{\sigma_{d}(P_{0})},\xi_{Q}(q_{0})\rangle,
    \end{equation*}
    we deduce
    \begin{equation*}
        \hat{J}_{d}(\xi)(P_{1})=\sigma_{d}(P_{0})\langle  \eta^{-}(P_{0})(\widetilde{\xi}(P_{0})\rangle
    \end{equation*}
    and so we have proved that
    \begin{equation*}
        \hat{J}_{d}(\xi)(P_{1})=\sigma_{d}(P_{0}) \hat{J}_{d}(\xi)(P_{0}).
    \end{equation*}
\end{proof}


\section{Exact discrete Lagrangian for contact systems}

\subsection{The contact exponential map}

Given a contact regular Lagrangian $L:TQ\times \R\rightarrow\R$, consider the corresponding Lagrangian vector field $\xi_{L}$ and denote its flow by $\phi_{t}^{\xi_{L}}$.

Define the open subset $U_{h}$ of $TQ\times \R$ given by
\begin{equation*}
    U_{h}=\{(q_{0},\dot{q}_{0},z_{0})\in TQ \times \R \ | \ \phi_{t}^{\xi_{L}} \ \text{is defined for} \ t\in [0,h]\}
\end{equation*}
and let the \textit{contact exponential map} be defined by
\begin{equation}
	\begin{split}
		\text{exp}_{h}^{\xi_{L}}: \ & U_{h} \subseteq TQ \times \R \rightarrow Q \times Q \times \R \\
						& (q_{0},\dot{q}_{0},z_{0})\mapsto (q_{0},q_{1},z_{0}),
	\end{split}
\end{equation}
where $q_{1}=p_{Q}\circ \phi_{h}^{\xi_{L}}(q_{0},\dot{q}_{0},z_{0})$ and $p_{Q}:TQ\times \R \rightarrow Q$ is the projection onto $Q$ given by $p_{Q}(v_{q},z)=q$ for $v_{q}\in T_{q} Q$.

We will prove that the contact exponential map is a local diffeomorphism, using the fact that the non-holonomic exponential map, i.e., the exponential map of a non-holonomic system is a local embedding (see \cite{AMM,MMdDM2016}). In fact, to every regular contact system, one can associate a non-holonomic Lagrangian system on $T(Q\times \R)$ with nonlinear constraints.

Consider the singular Lagrangian function
\begin{equation}
    \widetilde{L}:T(Q\times \R) \rightarrow \R, \quad \widetilde{L}=L\circ \pi,
\end{equation}
where $\pi:T(Q\times \R) \rightarrow TQ\times \R$ is a projection onto $TQ\times \R$. Also, we take the non-linear constraints
\begin{equation}
    M_{L}=\{ (q,z,\dot{q},\dot{z})\in T(Q\times \R) \ | \ \dot{z}=L(q,\dot{q},z) \}.
\end{equation}
Observe that  $M_{L}$ is the zero level set of the real-valued function $\Phi:T(Q\times \R) \rightarrow \R$ given by $\Phi(q,z,\dot{q},\dot{z})=\dot{z}-L(q,\dot{q},z)$.

The pair $(\widetilde{L},M_{L})$ forms a Lagrangian non-holonomic system with non-linear constraints determined by the submanifold $M_{L}$ and dynamics given by Chetaev's principle (see \cite{Bloch,LMdD1996} and references therein). According to this principle the equations of motion are
\begin{equation}\label{nhcontact}
    \begin{split}
        & \frac{\dd}{\dd t}\left( \frac{\partial \widetilde{L}}{\partial \dot{q}^{i}} \right) - \frac{\partial \widetilde{L}}{\partial q^{i}}=\lambda\frac{\partial \Phi}{\partial \dot{q}^{i}} \\
        & \frac{\dd}{\dd t}\left( \frac{\partial \widetilde{L}}{\partial \dot{z}} \right) - \frac{\partial \widetilde{L}}{\partial z}=\lambda\frac{\partial \Phi}{\partial \dot{z}} \\
        & \Phi(q^{i},z,\dot{q}^{i},\dot{z})=0,
    \end{split}
\end{equation}
with Lagrange multiplier $\lambda$. As $\widetilde{L}$ does not depend on $\dot{z}$ it is straightforward to check that the Lagrange multiplier is just
\begin{equation*}
    \lambda=-\frac{\partial L}{\partial z}
\end{equation*}
and that equations \eqref{nhcontact} are equivalent to the Herglotz equations for $L$.

Moreover, since $L$ is regular, we can define a SODE vector field $\Gamma_{(\widetilde{L},M_{L})}$ $\in \mathfrak{X}(M_{L})$ as the unique vector field on $M_{L}$ whose integral curves satisfy equations \eqref{nhcontact}. Hence, we deduce
\begin{equation}\label{pirelated}
    T\pi(\Gamma_{(\widetilde{L},M_{L})})=\xi_{L}\circ \pi.
\end{equation}

Let us denote the flow of the vector field $\Gamma_{(\widetilde{L},M_{L})}$ by $\phi_{t}^{\Gamma_{(\widetilde{L},M_{L})}}:M_{L}\rightarrow M_{L}$.

Consider now the submanifold of $\M_{L}$ given by
\begin{equation*}
	M_{L,h}=\{(q_{0},\dot{q}_{0},z_{0},\dot{z}_{0})\in T(Q \times \R) \ | \ \phi_{t}^{\Gamma_{(\widetilde{L},M_{L})}} \ \text{is defined for} \ t\in [0,h]\}.
\end{equation*}
We define the non-holonomic exponential map to be
\begin{equation}
	\begin{split}
		\text{exp}_{h}^{\Gamma_{(\widetilde{L},M_{L})}}: \ M_{L,h}\subseteq \ M_{L} & \longrightarrow (Q \times \R) \times (Q \times \R) \\
						(q_{0},z_{0},\dot{q}_{0},\dot{z}_{0}) & \mapsto (q_{0},z_{0},q_{1},z_{1}),
	\end{split}
\end{equation}
where $(q_{1},z_{1})=\tau_{Q\times \R}\circ \phi_{h}^{\Gamma_{(\widetilde{L},M_{L})}}(q_{0},z_{0},\dot{q}_{0},\dot{z}_{0})$, with $\tau_{Q\times \R}:T(Q\times \R) \rightarrow Q\times \R$ the tangent bundle projection.

In \cite{AMM}  the authors prove that there is an open subset $N_{h}\subseteq M_{L,h}$ such that the non-holonomic exponential map $\text{exp}_{h}^{\Gamma_{(\widetilde{L},M_{L})}}|_{N_{h}}$ is a smooth embedding and, hence, a diffeomorphism into its image, which we will denote by $M_{d}$.

\begin{theorem}
 There exists an open set $V_{h}\subseteq U_{h}$ such that the contact exponential map $\text{exp}_{h}^{\xi_{L}}|_{V_{h}}$ is a diffeomorphism.
\end{theorem}

\begin{proof}
Let us consider the non-holonomic system $(\widetilde{L},M_{L})$ defined previously.

According to equation \eqref{pirelated}, the vector fields $\xi_{L}$ and $\Gamma_{(\widetilde{L},M_{L})}$ are $\pi$-related therefore, its flows satisfy
\begin{equation*}
    \pi \circ \phi_{t}^{\Gamma_{(\widetilde{L},M_{L})}}=\phi_{t}^{\xi_{L}} \circ \pi.
\end{equation*}
We remark that $\pi|_{M_{L}}$ is a diffeomorphism, since $M_{L}$ is diffeomorphic to the graph of the Lagrangian function $L$. As such, we can also write
\begin{equation*}
    \phi_{t}^{\Gamma_{(\widetilde{L},M_{L})}}=(\pi|_{M_{L}})^{-1} \circ \phi_{t}^{\xi_{L}} \circ \pi|_{M_{L}}.
\end{equation*}
Thus, we can write the non-holonomic exponential map in terms of the contact dynamics in the following way
\begin{equation*}
    \text{exp}_{h}^{\Gamma_{(\widetilde{L},M_{L})}}(q_{0},z_{0},\dot{q}_{0},\dot{z}_{0})=(q_{0},z_{0},q_{1},z_{1}),
\end{equation*}
with $(q_{1},z_{1})=\tau_{Q\times \R}\circ (\pi|_{M_{L}})^{-1} \circ \phi_{h}^{\xi_{L}} \circ \pi|_{M_{L}}(q_{0},z_{0},\dot{q}_{0},\dot{z}_{0})$ where $\dot{z}_0=L(q_0, \dot{q}_0, z_0)$.

Also note that $\tau_{Q\times \R}\circ (\pi|_{M_{L}})^{-1}=p_{Q\times \R}$, where
\begin{equation*}
    p_{Q\times \R}:TQ\times \R \rightarrow Q\times \R, \quad p_{Q\times \R}(v_{q},z)=(q,z).
\end{equation*}
In Diagram~\eqref{projections} we show the different projections we can define on the manifolds involved in this section.

\begin{equation}\label{projections}
    \begin{tikzcd}
                                                                              & T(Q \times \mathbb{R}) \arrow[ld, "\pi"'] \arrow[rd, "\tau_{Q\times \mathbb{R}}"] &                                                \\
TQ\times \mathbb{R} \arrow[rr, "p_{Q\times \mathbb{R}}"] \arrow[rd, "p_{Q}"'] &                                                                                   & Q\times \mathbb{R} \arrow[ld, "\text{pr}_{1}"] \\
                                                                              & Q                                                                                 &                                               
\end{tikzcd}
\end{equation}

With these projections we can also write the contact exponential map as
\begin{equation*}
    \text{exp}_{h}^{\xi_{L}}(q_{0},\dot{q}_{0},z_{0})=(q_{0},q_{1},z_{0}),
\end{equation*}
with $q_{1}=\text{pr}_{1}\circ p_{Q\times \R} \circ \phi_{h}^{\xi_{L}} (q_{0},\dot{q}_{0},z_{0})$. Hence, we can write it as
\begin{equation}\label{exponentialrelation}
    \text{exp}_{h}^{\xi_{L}} = \widetilde{\text{pr}}_{1} \circ \text{exp}_{h}^{\Gamma_{(\widetilde{L},M_{L})}}\circ (\pi|_{M_{L}})^{-1},
\end{equation}
    with
\begin{equation*}
    \begin{split}
        \widetilde{\text{pr}}_{1}:(Q\times \R) \times (Q \times \R) & \longrightarrow Q \times Q \times \R \\
        (q_{0},z_{0},q_{1},z_{1}) & \mapsto (q_{0}, \text{pr}_{1}(q_{1},z_{1}),z_{0}).
    \end{split}
\end{equation*}

Therefore, if $\widetilde{\text{pr}}_{1}|_{M_{d}}$ is a local diffeomorphism then, by equation \eqref{exponentialrelation}, the contact exponential map $\text{exp}_{h}^{\xi_{L}}|_{V_{h}}$ is a diffeomorphism if we choose $$V_{h}=\pi|_{M_{L}}(N_{h}),$$
where $N_{h}$ is the open subset where $\text{exp}_{h}^{\Gamma_{(\widetilde{L},M_{L})}}|_{N_{h}}$ is an embedding.

We are going to prove in the next Lemma that $\widetilde{\text{pr}}_{1}|_{M_{d}}$ is a local diffeomorphism.

\end{proof}

\begin{lemma}
Using the same notation as in the previous theorem, $\widetilde{\text{pr}}_{1}|_{M_{d}}$ is a local diffeomorphism.
\end{lemma}

\begin{proof}
    All we must prove is that $\widetilde{\text{pr}}_{1}|_{M_{d}}$ is a local submersion (immersion) since, by dimensional reasons, this forces $\widetilde{\text{pr}}_{1}|_{M_{d}}$ to be also a local immersion (submersion).
    
    Let $x\in M_{d}$. Any vector in the kernel of $T_{x}\widetilde{\text{pr}}_{1}|_{M_{d}}$ must be the tangent vector of a curve of the form
    \begin{equation*}
        Z(s)=(q_{0},z_{0},q_{1},w \cdot s)\in M_{d}, \quad w\in \R.
    \end{equation*}
    Let $\gamma_{s}(t)=\phi_{t}^{\Gamma_{(\widetilde{L},M_{L})}}\circ (\text{exp}_{h}^{\Gamma_{(\widetilde{L},M_{L})}})^{-1}(Z(s))$. For each fixed value of $s$, this is an integral curve of $\Gamma_{(\widetilde{L},M_{L})}$ satisfying
    \begin{equation*}
        \tau_{Q\times \R} \circ \gamma_{s}(0)=(q_{0},z_{0}), \quad \tau_{Q\times \R} \circ \gamma_{s}(h)=(q_{1},w \cdot s).
    \end{equation*}
    
    Moreover, note that the projection of $\gamma_{s}(t)$ to $TQ\times \R$, i.e., the curve $\pi \circ \gamma_{s}(t)$ is an integral curve of $\xi_{L}$ with endpoints $q_{0}$ and $q_{1}$ for each fixed value of $s$ and so $\pi \circ \gamma_{0}(t)$ must satisfy Herglotz' principle. Note that the action over the curves $\pi \circ \gamma_{s}(t)$ is given by
    \begin{equation*}
        \mathcal{A}(p_{Q}\circ \pi \circ \gamma_{s}(t))=p_{\R} \circ \pi \circ \gamma_{s}(h)=w \cdot s,
    \end{equation*}
    where $p_{\R}:TQ\times \R \rightarrow \R$ is the projection onto the second factor.
    
    Therefore, $p_{Q}\circ \pi \circ \gamma_{0}(t)$ is a critical value of the action if and only if $w=0$. Therefore, $T_{x}\widetilde{\text{pr}}_{1}|_{M_{d}}$ is trivial and $\widetilde{\text{pr}}_{1}|_{M_{d}}$ must be a local diffeomorphism in a neighbourhood of each point.
\end{proof}

Since the contact exponential map is a local diffeomorphism we can define a local inverse called the \textit{exact retraction} and denote it by $R_{h}^{e-}:Q \times Q \times \R \rightarrow TQ \times \R$. We will also use its translation by the flow
\begin{equation*}
	R_{h}^{e+}:Q \times Q \times \R \rightarrow TQ \times \R, \quad R_{h}^{e+}:=\phi_{h}^{\xi_{L}}\circ R_{h}^{e-}.
\end{equation*}

\subsection{The exact discrete Lagrangian function}

Let $L_{h}^{e}: Q \times Q \times \R \rightarrow \R$ and defined by
\begin{equation}
	L_{h}^{e} (q_{0},q_{1},z_{0})=\int_{0}^{h} L\circ \phi_{t}^{\xi_{L}}\circ R_{h}^{e-}(q_{0},q_{1},z_{0}) \dd t
\end{equation}
is called the \textit{exact discrete Lagrangian} function.

We will need the following classical result in the proof of the next theorem: the solution of the first order linear equation $\dot{y}=a(t)+\frac{\dd b}{\dd t}(t)y$ is
\begin{equation}\label{LODE}
y(t)=e^{b(t)}\left( \int_{0}^{t} a(s)e^{-b(s)} \ \dd s + y(0) \right).
\end{equation}

\begin{theorem}
The Legendre transforms of a regular Lagrangian $L:TQ\times\R\rightarrow\R$ are related to the discrete Legendre transforms of the corresponding exact discrete Lagrangian $L_{h}^{e}:Q\times Q\times \R\rightarrow\R$ in the following way
	\begin{equation}
		\F^{+}L_{h}^{e}=\F L\circ R_{h}^{e+}, \quad \F^{-}L_{h}^{e}=\F L\circ R_{h}^{e-}.
	\end{equation}
\end{theorem}

\begin{proof}
	
	We will prove in local computations that the derivatives of the exact discrete Lagrangian function satisfy
	\begin{equation}
	\begin{split}
	& D_{1} L_{h}^{e} (q_{0},q_{1},z_{0}) = -\frac{\partial L}{\partial \dot{q}}(q_{0},\dot{q}_{0},z_{0})e^{b(h)}; \\
	& D_{2} L_{h}^{e} (q_{0},q_{1},z_{0}) = \frac{\partial L}{\partial \dot{q}}(q_{1},\dot{q}_{1},z_{1}); \\
	& D_{z}  L_{h}^{e} (q_{0},q_{1},z_{0}) = e^{b(h)}-1.
	\end{split}
	\end{equation}
	where
	\begin{equation}
	\begin{split}
	& (q_{0},\dot{q}_{0},z_{0})=R_{h}^{e-}(q_{0},q_{1},z_{0}), \quad (q_{1},\dot{q}_{1},z_{1})= \phi_{h}^{\xi_{L}}\circ R_{h}^{e-}(q_{0},q_{1},z_{0}), \\
	& \text{and} \quad b(t)=\int_{0}^{t} \frac{\partial L}{\partial z}(\phi_{s}^{\xi_{L}}\circ R_{h}^{e-}(q_{0},q_{1},z_{0})) \ \dd s
	\end{split}
	\end{equation}
	
To simplify the notation in the proof we will use the notation $\gamma_{0,1}(t)=(q_{0,1}(t),\dot{q}_{0,1}(t),z_{0,1}(t)):=\phi_{t}^{\xi_{L}}\circ R_{h}^{e-}(q_{0},q_{1},z_{0})$. Under this convention we will have
\begin{equation*}
	L_{h}^{e} (q_{0},q_{1},z_{0})=\int_{0}^{h} L(\gamma_{0,1}(t)) \dd t.
\end{equation*}
Note first that any variation of the exact discrete Lagrangian will take the form
\begin{equation}
\begin{split}
& \delta L_{h}^{e} (q_{0},q_{1},z_{0}) = \left. \frac{\dd}{\dd s} \right|_{s=0}L_{h}^{e}(q_{0}(s),q_{1}(s),z_{0}(s)) \\
& =\int_{0}^{h}\frac{\partial L}{\partial q}(\gamma_{0,1}(t))\delta q_{0,1} +\frac{\partial L}{\partial \dot{q}}(\gamma_{0,1}(t))\delta \dot{q}_{0,1}+\frac{\partial L}{\partial z}(\gamma_{0,1}(t))\delta z_{0,1} \ \dd t.
\end{split}
\end{equation}
Since $\gamma_{0,1}(t)$ is a solution of Euler-Lagrange equations, it satisfies
\begin{equation*}
\dot{z}_{0,1}=L(q_{0,1}(t),\dot{q}_{0,1}(t),z_{0,1}(t)).
\end{equation*}
Therefore, any variation of $z_{0,1}$ satisfies the variational equation
\begin{equation}\label{ode}
\delta \dot{z}_{0,1}=\frac{\partial L}{\partial q}(\gamma_{0,1}(t))\delta q_{0,1} +\frac{\partial L}{\partial \dot{q}}(\gamma_{0,1}(t))\delta \dot{q}_{0,1}+\frac{\partial L}{\partial z}(\gamma_{0,1}(t))\delta z_{0,1}.
\end{equation}
Hence, any variation of the exact discrete Lagrangian reduces to
\begin{equation}
\delta L_{h}^{e} (q_{0},q_{1},z_{0}) = \delta z_{0,1}(h)-\delta z_{0,1}(0).
\end{equation}
Moreover, we can solve the function $\delta z_{0,1}$ explicitly, by solving the differential equation \eqref{ode}
\begin{equation}
\delta z_{0,1}(h)=e^{b(h)}\left(\int_{0}^{h} a(s)e^{-b(s)} \ \dd s + \delta z_{0,1}(0) \right),
\end{equation}
with
\begin{equation*}
	\begin{split}
b(t) & = \int_{0}^{t} \frac{\partial L}{\partial z}(\gamma_{0,1}(s)) \ \dd s, \\
a(t) & = \frac{\partial L}{\partial q}(\gamma_{0,1}(t))\delta q_{0,1} +\frac{\partial L}{\partial \dot{q}}(\gamma_{0,1}(t))\delta \dot{q}_{0,1}.
	\end{split}
\end{equation*}
Let us compute the integration in the expression of $\delta z_{0,1}$:
\begin{equation*}
	\begin{split}
		\int_{0}^{h} a(s)e^{-b(s)} \ \dd s & = \int_{0}^{h} \left( \frac{\partial L}{\partial q}\delta q_{0,1} +\frac{\partial L}{\partial \dot{q}}\delta \dot{q}_{0,1} \right)e^{-b(t)} \ \dd t \\
		& = \int_{0}^{h} \left( \frac{\partial L}{\partial q} - \frac{d}{dt}\frac{\partial L}{\partial \dot{q}}+\frac{\partial L}{\partial \dot{q}}\frac{\partial L}{\partial z}\right)\delta q_{0,1} e^{-b(t)} \ \dd t \\
		& + \frac{\partial L}{\partial \dot{q}}(\gamma_{0,1}(h))e^{-b(h)}\delta q_{0,1}(h)-\frac{\partial L}{\partial \dot{q}}(\gamma_{0,1}(0))\delta q_{0,1}(0),
	\end{split}
\end{equation*}
where we are using integration by parts. Note that the term between brackets is zero, since we are over solutions of Euler-Lagrange equations. Therefore,
\begin{equation}
\delta z_{0,1}(h)=\frac{\partial L}{\partial \dot{q}}(\gamma_{0,1}(h))\delta q_{0,1}(h)-\frac{\partial L}{\partial \dot{q}}(\gamma_{0,1}(0))e^{b(h)}\delta q_{0,1}(0) + e^{b(h)}\delta z_{0,1}(0).
\end{equation}

Note that the differentials of the discrete Lagrangian $D_{1}L_{h}^{e}$, $D_{2}L_{h}^{e}$ and $D_{z} L_{h}^{e}$ are instances of particular variations. Therefore, we have that
\begin{equation}
	\begin{split}
D_{1} L_{h}^{e}(q_{0},q_{1},z_{0}) & = \left(\frac{\partial L}{\partial \dot{q}}(\gamma_{0,1}(h))\frac{\partial q_{0,1}(h)}{\partial q_{0}^{i}} -\frac{\partial L}{\partial \dot{q}}(\gamma_{0,1}(0))e^{b(h)}\frac{\partial q_{0,1}(0)}{\partial q_{0}^{i}} \right. \\
									& \left.+ (e^{b(h)}-1)\frac{\partial z_{0,1}(0)}{\partial q_{0}^{i}} \right) \dd q_{0}^{i}\\
									& = -\frac{\partial L}{\partial \dot{q}^{i}}(\gamma_{0,1}(0))e^{b(h)}\dd q_{0}^{i},
	\end{split}
\end{equation}
since $q_{0,1}(h)\equiv q_{1}$ and so its derivative with respect to $q_{0}$ vanishes, $q_{0,1}(0)\equiv q_{0}$ and so its derivative with respect to $q_{0}$ is the identity and, finally, $z_{0,1}(0)\equiv z_{0}$ does not depend upon $q_{0}$. Likewise, the next derivative follows from applying similar arguments. Indeed, we have that
\begin{equation}
	\begin{split}
D_{2} L_{h}^{e}(q_{0},q_{1},z_{0}) & = \left(\frac{\partial L}{\partial \dot{q}}(\gamma_{0,1}(h))\frac{\partial q_{0,1}(h)}{\partial q_{1}^{i}} -\frac{\partial L}{\partial \dot{q}}(\gamma_{0,1}(0))e^{b(h)}\frac{\partial q_{0,1}(0)}{\partial q_{1}^{i}} \right. \\
									& \left.+ (e^{b(h)}-1)\frac{\partial z_{0,1}(0)}{\partial q_{1}^{i}} \right) \dd q_{1}^{i}\\
									& = \frac{\partial L}{\partial \dot{q}^{i}}(\gamma_{0,1}(h))\dd q_{1}^{i}.
	\end{split}
\end{equation}
Analogously, we also deduce
\begin{equation}
	\begin{split}
D_{z} L_{h}^{e}(q_{0},q_{1},z_{0}) & = \left(\frac{\partial L}{\partial \dot{q}}(\gamma_{0,1}(h))\frac{\partial q_{0,1}(h)}{\partial z_{0}} -\frac{\partial L}{\partial \dot{q}}(\gamma_{0,1}(0))e^{b(h)}\frac{\partial q_{0,1}(0)}{\partial z_{0}} \right. \\
									& \left.+ (e^{b(h)}-1)\frac{\partial z_{0,1}(0)}{\partial z_{0}} \right) \dd z_{0}\\
									& = (e^{b(h)}-1)\dd z_{0}.
	\end{split}
\end{equation}
Now, the result follows by the definition of the discrete Legendre transforms.
\end{proof}

The commutativity of the following diagram summarizes the statement of the previous theorem
\begin{equation}
		\begin{tikzcd}
		Q\times Q\times \mathbb{R} \arrow[r, "R^{e\pm}_{h}"] \arrow[rd, "\mathbb{F} ^{\pm} L^{e}_{h}"'] & TQ\times \mathbb{R} \arrow[d, "\mathbb{F} L"] \\
		& T^{*}Q\times \mathbb{R}                      
		\end{tikzcd}
\end{equation}

Now, we are going to relate the continuous contact Lagrangian flow with its discrete counterpart, when we take as discrete Lagrangian the corresponding exact discrete Lagrangian.

\begin{theorem}
Take a regular Lagrangian $L:TQ\rightarrow \R$ and fix a time step $h>0$. Then we have that:
\begin{enumerate}
 \item $L_{h}^{e}$ is a regular discrete Lagrangian function;
 \item If $H$ is the Hamiltonian function corresponding to $L$ introduced at the end of Section \ref{contact:lagrangian} and $\phi_{t}^{X_{H}}$ is its contact Hamiltonian flow, we have that
 \begin{equation} 
 \F^{+} L_{h}^{e}=\phi_{h}^{X_{H}}\circ \F^{-} L_{h}^{e}.
 \end{equation}
 \item If $(q,z):[0,Nh]\rightarrow Q\times \R$ is a solution of the Herglotz equations, then it is related to the solution of the discrete Herglotz equations $\{ (q_{0},z_{0}),(q_{1},z_{1}),...,(q_{N},z_{N}) \}$ for the corresponding exact discrete Lagrangian with $(q(0),q(h),z(0))$ as initial conditions in the following way:
 	\begin{equation}
		q_{k}=q(kh), \quad z_{k}=z(kh) \quad \text{for} \ k=0,..., N.
	\end{equation}
\end{enumerate}
\end{theorem}

\begin{proof}
	Item 1. is a consequence of of the theorem before, since $\F^{-}L_{h}^{e}$ is a composition of two local diffeomorphisms it is itself a local diffeomorphism. Item 2. comes from unyielding the definitions:
	\begin{equation*}
		\F^{+} L_{h}^{e}=\F L \circ R^{e+}_{h}=\F L \circ \phi_{h}^{\Gamma_{L}}\circ R^{e-}_{h}=\phi_{h}^{X_{H}}\circ\F L\circ R^{e-}_{h}=\phi_{h}^{X_{H}}\circ \F^{-} L_{h}^{e}.
	\end{equation*}
	For item 3. observe that, by discrete Herglotz equations, for every $k=1,...,N-1$ we have that
	\begin{equation*}
		\F^{+} L_{h}^{e}(q(k-1),q(k),z(k-1))=\F^{-} L_{h}^{e}(q(k),q(k+1),z(k))
	\end{equation*}
	so that $\{ (q_{0},z_{0}),(q_{1},z_{1}),...,(q_{N},z_{N}) \}$ is indeed the solution of these equations.
\end{proof}


\section{Numerical examples}

Given a mechanical contact Lagrangian with a euclidean metric and a potential function $V:Q\rightarrow \R$ of the type
\begin{equation*}
L(q,\dot{q},z)=\frac{1}{2}\dot{q}^{2}-V(q)+\gamma z, \quad (q,\dot{q},z)\in TQ\times \R, \quad \gamma < 0.
\end{equation*}
one usually approximates the exact discrete Lagrangian associated to $L$ by means of a quadrature rule. Note that the restriction of $\gamma$ to negative values is necessary to model dissipative dynamics, though we could define the integrator for any value of $\gamma\in \R$. If we use the middle point rule to approximate the positions, i.e., $q\approx\frac{q_{1}+q_{0}}{2}$, one may define the discrete Lagrangian $L_{d}:Q\times Q\times \R\rightarrow \R$ in the following way
\begin{equation*}
L_{d}(q_{0},q_{1},z_{0})=\frac{1}{2h}(q_{1}-q_{0})^{2}-hV\left(\frac{q_{1}+q_{0}}{2}\right)+h\gamma z_{0}.
\end{equation*}
We remark that the value of $h$ should be chosen small enough so that the function $\sigma_{d}$ does not vanish anywhere. In this case, the discrete Herglotz equations are of the type
\begin{equation*}
    \begin{split}
        & \frac{q_{1}-q_{0}}{h}-\frac{h}{2}\frac{\partial V}{\partial q}\left(\frac{q_{1}+q_{0}}{2}\right)=\frac{1}{(1+h\gamma)}\left(\frac{q_{2}-q_{1}}{h}+\frac{h}{2}\frac{\partial V}{\partial q}\left(\frac{q_{2}+q_{1}}{2}\right)\right)  \\
			& z_{1}=L_{d}(q_{0},q_{1},z_{0})=\frac{1}{2h}(q_{1}-q_{0})^{2}-hV\left(\frac{q_{1}+q_{0}}{2}\right)+(h\gamma+1) z_{0}
    \end{split}
\end{equation*}

\begin{example}
	The free single particle contact Lagrangian is
	\begin{equation*}
		L(q,\dot{q},z)=\frac{1}{2}\dot{q}^{2}+\gamma z, \quad (q,\dot{q},z)\in TQ\times \R.
	\end{equation*}
	A simple discretization of this Lagrangian would be
	\begin{equation}\label{eq:free_particle_discrete}
	L_{d}(q_{0},q_{1},z_{0})=\frac{1}{2h}(q_{1}-q_{0})^{2}+h\gamma z_{0}.
	\end{equation}
	Then, choosing $h$ small enough so that the function $\sigma_{d}$ is non-vanishing, the discrete Herglotz equations for $L_{d}$ are locally given by
	\begin{equation*}
		\begin{split}
			& \frac{q_{1}-q_{0}}{h}=\frac{q_{2}-q_{1}}{h(1+h\gamma)} \quad \Rightarrow \quad q_{2}=(h\gamma+2)q_{1}-(h\gamma+1)q_{0} \\
			& z_{1}=\frac{1}{2h}(q_{1}-q_{0})^{2}+(h\gamma+1) z_{0}
		\end{split}		
	\end{equation*}

	The discrete flow obtained by solving these equations is plotted in Fig.~\ref{fig:free_particle}.

	In this case, one can also compute the exact discrete Lagrangian and solve the exact dynamics.
	\begin{equation}
		L_{h}^{e}(q_{0},q_{1},z_{0}) = \frac{\gamma \left(q_{1}- q_{0}\right)^{2} e^{\gamma h}}{2 e^{\gamma h} - 2} - z_{0} \left(e^{\gamma h} - 1\right).
	\end{equation}

	\begin{figure}[htb!]
		\centering
		\input{fig/free_particle_qt.tex}\hfill
		\input{fig/free_particle_zt.tex}\\
		\input{fig/free_particle_logH.tex}
		\caption{Position $q$ and $z$ and logarithm of the discrete Hamiltonian $H\circ \F^- L_d$ for a free particle, computed by solving the discrete Herglotz equations for the discrete Lagrangian~\eqref{eq:free_particle_discrete} (continuous line) and the exact dynamics (dashed line), for $\gamma = -0.05$ and the time-step $h = 0.5$. The initial conditions are $q_{0} = 1$, $q_{1} = 2$ and $z_0 = 0$.}
		\label{fig:free_particle}
	\end{figure}
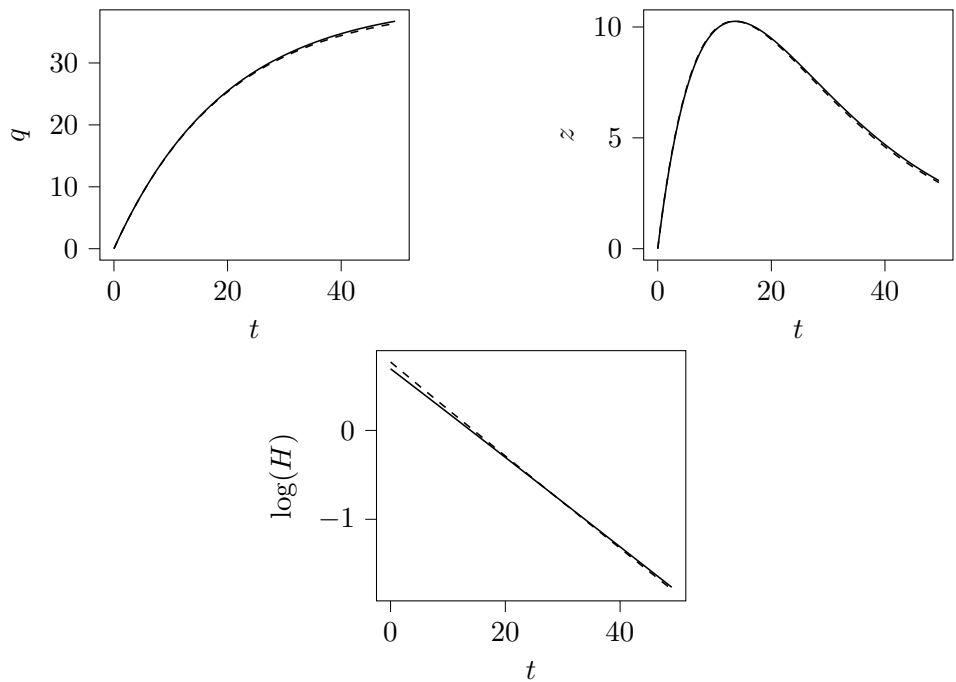
\end{example}

\begin{example}
	The damped harmonic oscillator is described by the Lagrangian
	\begin{equation*}
		L(q,\dot{q},z)=\frac{1}{2}\dot{q}^{2}-\frac{1}{2}q^{2}+\gamma z, \quad (q,\dot{q},z)\in TQ\times \R.
	\end{equation*}
	Using a middle point discretization, i.e., $q\approx\frac{q_{1}+q_{0}}{2}$, one may define the discrete Lagrangian
	\begin{equation}\label{eq:harmonic_oscillator_discrete}
	L_{d}(q_{0},q_{1},z_{0})=\frac{1}{2h}(q_{1}-q_{0})^{2}-\frac{h}{8}(q_{1}+q_{0})^{2}+h\gamma z_{0}.
	\end{equation}
	In this case, after choosing $h$ small enough, the discrete Herglotz equations hold
	\begin{equation*}
		\begin{split}
			& \frac{q_{1}-q_{0}}{h}-\frac{h}{4}(q_{1}+q_{0})=\frac{1}{(1+h\gamma)}\left(\frac{q_{2}-q_{1}}{h}+\frac{h}{4}(q_{2}+q_{1})\right)  \\
			& z_{1}=\frac{1}{2h}(q_{1}-q_{0})^{2}-\frac{h}{8}(q_{1}+q_{0})^{2}+(h\gamma+1) z_{0},
		\end{split}		
	\end{equation*}
	which can be solved explicitly for $q_{2}$
	\begin{equation*}
	    q_{2}=-\frac{(h^{3}\gamma+4h\gamma+h^{2}+4)q_{0}+(h^{3}\gamma-4h\gamma+2h^{2}-8)q_{1}}{h^{2}+4}.
	\end{equation*}

	The discrete flow obtained by solving these equations is plotted in Fig.~\ref{fig:harmonic_oscillator}.

	In this case, the exact discrete Lagrangian and the exact discrete dynamics can be computed with the aid of a Computer Algebra system, but the analytic expressions are complicated, so we only include their graph in Fig.~\ref{fig:harmonic_oscillator}.
	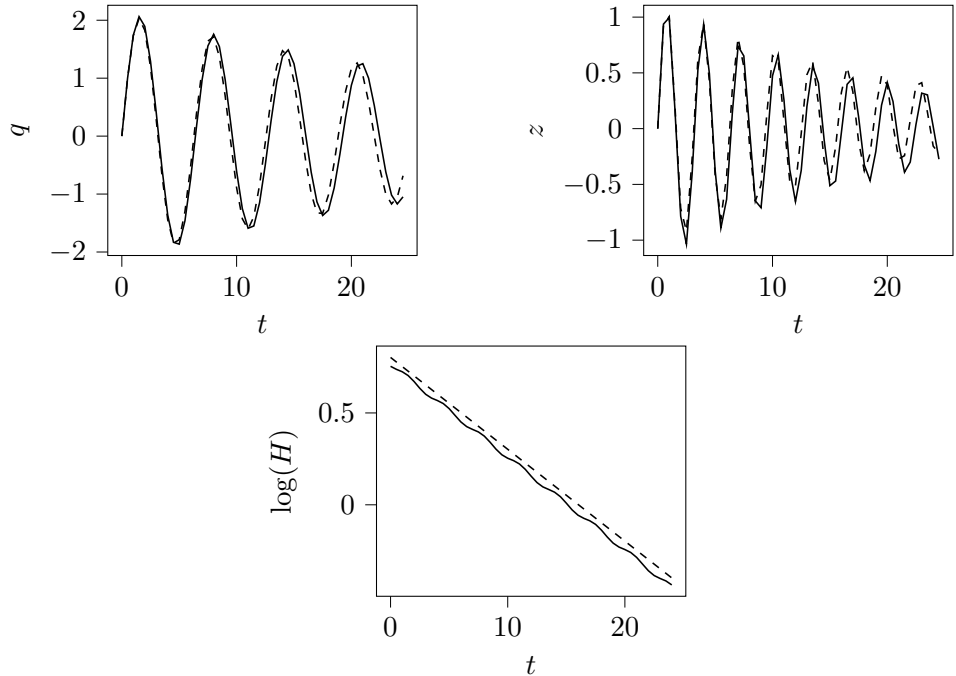
\begin{figure}[htb!]
		\centering
		\input{fig/harmonic_oscillator_qt.tex}\hfill
		\input{fig/harmonic_oscillator_zt.tex}\\ 
		\input{fig/harmonic_oscillator_logH.tex}
		\caption{Position $q$ and $z$ and logarithm of the discrete Hamiltonian $H\circ \F^- L_d$ for a harmonic oscillator, computed by solving the discrete Herglotz equations on the discrete Lagrangian~\eqref{eq:harmonic_oscillator_discrete} (continuous line) and the exact dynamics (dashed line), for $\gamma =- 0.05$ and the time-step $h = 0.5$. The initial conditions are $
		q_{0} = 1$, $q_{1} = 2$ and $z_0 = 0$.}
		\label{fig:harmonic_oscillator}
	\end{figure}
\end{example}


\section{Conclusions and future work}
In this paper, we went deeper in the geometry of discrete contact mechanics following, as a starting point, the results by \cite{VBS}. We have done a detailed study of the discrete Herglotz principle and its geometric properties, including the discrete Legendre transforms and the associated discrete Lagrangian and Hamiltonian flows. Moreover, we have analyzed the existence of dissipated quantities related with symmetries of the system and the construction of the exact discrete Lagrangian function given the correspondence between the discrete and continuous system.  

In future work, we will study the variational error analysis allowing us to estimate the error order of the proposed methods just from the error of approximation of the exact discrete Lagrangian function, that is, how well the discrete Lagrangian function matches the exact discrete Lagrangian function \cite{marsden-west,PatrickCuell}. Moreover, we will introduce higher-order methods for contact Lagrangian systems extending the theory of Morse functions to Legendrian submanifolds (see \cite{marle,BCGM,FLMMV}). For instance, this theory will give a complete geometric explanation of other possible discretizations of the phase space, as for instance,  the one used by
Vermeeren {\it et al} which is $Q \times Q \times \mathbb{R}^2$ instead of $Q\times Q\times \R$.

 \noindent \section*{Acknowledgments}Manuel Lainz wishes to thank MICINN and ICMAT for a FPI-Severo Ochoa predoctoral contract PRE2018-083203

The authors are  supported  by Ministerio de Ciencia e Innovaci\'on ( Spain) under grants  MTM2016-76702-P and ``Severo Ochoa Programme for Centres of Excellence'' in R\&D (SEV-2015-0554). A.Simoes is supported by the FCT (Portugal) research fellowship SFRH/BD/129882/2017. M.~Lainz wishes to thank MICINN and ICMAT for a FPI-Severo Ochoa predoctoral contract PRE2018-083203.

\bibliography{Thermo_Discrete}{}
\bibliographystyle{apalike}

\end{document}

%% file: fig/free_particle_qt.tex
\begin{tikzpicture}

\begin{axis}[
tick align=outside,
tick pos=left,
width=0.45\textwidth,
x grid style={white!69.0196078431373!black},
xlabel={\(\displaystyle t\)},
xmin=-2.475, xmax=51.975,
xtick style={color=black},
y grid style={white!69.0196078431373!black},
ylabel={\(\displaystyle q\)},
ymin=-1.83688761054016, ymax=38.5746398213434,
ytick style={color=black}
]
\addplot [semithick, black]
table {%
0 0
0.5 1
1 1.975
1.5 2.925625
2 3.852484375
2.5 4.756172265625
3 5.63726795898438
3.5 6.49633626000977
4 7.33392785350953
4.5 8.15057965717179
5 8.9468151657425
5.5 9.72314478659895
6 10.480066166934
6.5 11.2180645127606
7 11.9376128999416
7.5 12.6391725774431
8 13.323193263007
8.5 13.9901134314319
9 14.6403605956461
9.5 15.2743515807549
10 15.8924927912361
10.5 16.4951804714552
11 17.0828009596688
11.5 17.6557309356771
12 18.2143376622852
12.5 18.7589792207281
13 19.2900047402099
13.5 19.8077546217047
14 20.3125607561621
14.5 20.8047467372581
15 21.2846280688267
15.5 21.7525123671061
16 22.2086995579285
16.5 22.6534820689803
17 23.0871450172558
17.5 23.5099663918245
18 23.9222172320289
18.5 24.3241618012283
19 24.7160577561976
19.5 25.0981563122927
20 25.4707024044855
20.5 25.8339348443734
21 26.1880864732641
21.5 26.5333843114326
22 26.8700497036469
22.5 27.1982984610558
23 27.5183409995295
23.5 27.8303824745413
24 28.1346229126779
24.5 28.431257339861
25 28.7204759063646
25.5 29.0024640087056
26 29.277402408488
26.5 29.5454673482759
27 29.8068306645691
27.5 30.061659897955
28 30.3101184005063
28.5 30.5523654404937
29 30.7885563044815
29.5 31.0188423968696
30 31.243371336948
30.5 31.4622870535245
31 31.6757298771865
31.5 31.883836630257
32 32.0867407145007
32.5 32.2845721966384
33 32.4774578917226
33.5 32.6655214444297
34 32.8488834083191
34.5 33.0276613231113
35 33.2019697900337
35.5 33.371920545283
36 33.5376225316511
36.5 33.69918196836
37 33.8567024191512
37.5 34.0102848586726
38 34.160027737206
38.5 34.3060270437761
39 34.4483763676819
39.5 34.58716695849
40 34.722487784528
40.5 34.854425589915
41 34.9830649501674
41.5 35.1084883264134
42 35.2307761182533
42.5 35.3500067152972
43 35.466256547415
43.5 35.5796001337299
44 35.6901101303869
44.5 35.7978573771275
45 35.9029109426996
45.5 36.0053381691323
46 36.1052047149043
46.5 36.202574597032
47 36.2975102321065
47.5 36.3900724763041
48 36.4803206643968
48.5 36.5683126477871
49 36.6541048315928
49.5 36.7377522108032
};
\addplot [semithick, black, dashed]
table {%
0 0
0.5 1
1 1.97468487947557
1.5 2.92469549375388
2 3.85065647483224
2.5 4.75317664207368
3 5.63284940250567
3.5 6.49025314098526
4 7.32595160048713
4.5 8.14049425276464
5 8.93441665962744
5.5 9.70824082507345
6 10.4624755385065
6.5 11.1976167092652
7 11.9141476926837
7.5 12.6125396078975
8 13.2932516476044
8.5 13.9567313799836
9 14.6034150429721
9.5 15.2337278310908
10 15.8480841750103
10.5 16.4468880140384
11 17.030533061711
11.5 17.5994030646584
12 18.1538720549183
12.5 18.6943045958628
13 19.2210560218979
13.5 19.7344726720964
14 20.234892117916
14.5 20.7226433851518
15 21.1980471702716
15.5 21.6614160512733
16 22.1130546932052
16.5 22.553260048483
17 22.9823215521364
17.5 23.4005213121125
18 23.8081342947614
18.5 24.2054285056272
19 24.5926651656613
19.5 24.9700988829751
20 25.3379778202451
20.5 25.6965438578797
21 26.0460327530556
21.5 26.3866742947281
22 26.7186924547176
22.5 27.0423055349705
23 27.3577263110936
23.5 27.6651621722532
24 27.9648152575339
24.5 28.2568825888452
25 28.5415562004631
25.5 28.8190232652928
26 29.0894662179347
26.5 29.3530628746354
27 29.6099865502019
27.5 29.8604061719558
28 30.1044863908034
28.5 30.3423876894931
29 30.5742664881335
29.5 30.8002752470393
30 31.0205625669738
30.5 31.235273286854
31 31.4445485789827
31.5 31.6485260418682
32 31.8473397906965
32.5 32.0411205455113
33 32.2299957171626
33.5 32.4140894910794
34 32.5935229089217
34.5 32.7684139481652
35 32.9388775996715
35.5 33.1050259432949
36 33.2669682215745
36.5 33.4248109115615
37 33.5786577948274
37.5 33.7286100257012
38 33.8747661977774
38.5 34.0172224087421
39 34.1560723235568
39.5 34.2914072360431
40 34.4233161289086
40.5 34.5518857322529
41 34.6772005805927
41.5 34.7993430684433
42 34.9183935044928
42.5 35.0344301644051
43 35.1475293422865
43.5 35.2577654008485
44 35.3652108203019
44.5 35.469936246012
45 35.5720105349483
45.5 35.6715008009576
46 35.7684724588918
46.5 35.862989267618
47 35.9551133719396
47.5 36.0449053434571
48 36.1324242203935
48.5 36.217727546412
49 36.3008714084512
49.5 36.381910473602
};
\end{axis}

\end{tikzpicture}

%% file: fig/free_particle_zt.tex
\begin{tikzpicture}

\begin{axis}[
tick align=outside,
tick pos=left,
width=0.45\textwidth,
x grid style={white!69.0196078431373!black},
xlabel={\(\displaystyle t\)},
xmin=-2.475, xmax=51.975,
xtick style={color=black},
y grid style={white!69.0196078431373!black},
ylabel={\(\displaystyle z\)},
ymin=-0.512959232079956, ymax=10.7721438736791,
ytick style={color=black}
]
\addplot [semithick, black]
table {%
0 0
0.5 1
1 1.925625
1.5 2.781172265625
2 3.57071126000977
2.5 4.29809528217178
3 4.96697252097393
3.5 5.58079655377623
4 6.14283631743329
4.5 6.65618557792229
5 7.12377192358309
5.5 7.54836530571262
6 7.9325861490781
6.5 8.27891305379402
7 8.58969010894394
7.5 8.8671338373163
8 9.11333978966275
8.5 9.33028880597299
9 9.5198529603923
9.5 9.68380120558179
10 9.82380473153734
10.5 9.9414420531368
11 10.0382038399768
11.5 10.1154975013863
12 10.1746515388634
12.5 10.2169196775749
13 10.2434847879765
13.5 10.255462608065
14 10.2539052762492
14.5 10.2398046843303
15 10.2140956596101
15.5 10.1776589846962
16 10.1313242631493
16.5 10.0758726387081
17 10.0120393754474
17.5 9.94051630585337
18 9.86195415345631
18.5 9.77696473632874
19 9.68612305744186
19.5 9.58996928757579
20 9.48901064619448
20.5 9.38372318542656
21 9.27455348203686
21.5 9.16192024202976
22 9.04621582229379
22.5 8.92780767347694
23 8.8070397080727
23.5 8.68423359749845
24 8.55969000175851
24.5 8.43368973510482
25 8.30649487093758
25.5 8.17834978902601
26 8.0494821679753
26.5 7.92010392571941
27 7.79041211068021
27.5 7.66058974610124
28 7.53080662993872
28.5 7.40122009257294
29 7.27197571448991
29.5 7.14320800597504
30 7.01504105075839
30.5 6.88758911545361
31 6.7609572265401
31.5 6.63524171655014
32 6.51053074103917
32.5 6.38690476783797
33 6.26443704001014
33.5 6.1431940138667
34 6.02323577332142
34.5 5.90461642180583
35 5.78738445290152
35.5 5.6715831007888
36 5.55725067155542
36.5 5.44442085635623
37 5.33312302736479
37.5 5.22338251741002
38 5.11522088414623
38.5 5.00865615956151
39 4.90370308558891
39.5 4.80037333654607
40 4.69867572909201
40.5 4.59861642035505
41 4.50019909485231
41.5 4.40342514078996
42 4.30829381630329
42.5 4.21480240616715
43 4.12294636948039
43.5 4.0327194788021
44 3.94411395119318
44.5 3.85712057159352
45 3.77172880894309
45.5 3.68792692543423
46 3.60570207926279
46.5 3.52504042122678
47 3.44592718550311
47.5 3.36834677491643
48 3.29228284099753
48.5 3.21771835911356
49 3.14463569893786
49.5 3.0730166905132
};
\addplot [semithick, black, dashed]
table {%
0 0
0.5 1.01255208279081
1 1.94885443093858
1.5 2.81336762123583
2 3.61031282201487
2.5 4.34368417688022
3 5.01726055885869
3.5 5.63461672670579
4 6.19913391351428
4.5 6.71400987625671
5 7.18226843345646
5.5 7.60676851681587
6 7.99021276133285
6.5 8.33515565720521
7 8.64401128565136
7.5 8.9190606596642
8 9.16245868965897
8.5 9.37624079297255
9 9.56232916521894
9.5 9.7225387306003
10 9.85858278741354
10.5 9.97207836417532
11 10.064551301013
11.5 10.1374410702324
12 10.1921053492722
12.5 10.2298243585925
13 10.2518049764105
13.5 10.2591846415991
14 10.2530350554937
14.5 10.2343656928101
15 10.204127131365
15.5 10.1632142098006
16 10.1124690220516
16.5 10.0526837568518
17 9.98460339016195
17.5 9.90892823799935
18 9.82631637677517
18.5 9.73738593788608
19 9.64271728296657
19.5 9.5428550658847
20 9.43831018725761
20.5 9.32956164697078
21 9.21705829990867
21.5 9.10122051984083
22 8.98244177615804
22.5 8.86109012791557
23 8.73750963941547
23.5 8.6120217213455
24 8.48492640128936
24.5 8.35650352722966
25 8.22701390748167
25.5 8.09670039032206
26 7.96578888641103
26.5 7.83448933694982
27 7.70299663036582
27.5 7.57149147017634
28 7.44014119654744
28.5 7.30910056393648
29 7.17851247708582
29.5 7.04850868752001
30 6.91921045258913
30.5 6.79072915899762
31 6.66316691265866
31.5 6.53661709662115
32 6.41116489872696
32.5 6.28688781057187
33 6.16385609926346
33.5 6.04213325339298
34 5.92177640456607
34.5 5.80283672576836
35 5.68535980777716
35.5 5.56938601476806
36 5.4549508202071
36.5 5.34208512406286
37 5.2308155523202
37.5 5.121164739727
38 5.01315159665742
38.5 4.90679156093004
39 4.80209683537612
39.5 4.69907661191234
40 4.59773728283374
40.5 4.49808264000553
41 4.40011406259789
41.5 4.30383069397412
42 4.20922960831177
42.5 4.11630596750559
43 4.02505316887373
43.5 3.93546298416083
44 3.84752569030663
44.5 3.76123019242429
45 3.67656413940951
45.5 3.59351403257982
46 3.51206532772254
46.5 3.43220253091031
47 3.35390928842446
47.5 3.27716847110847
48 3.2019622534574
48.5 3.12827218773277
49 3.05607927337754
49.5 2.98536402199118
};
\end{axis}

\end{tikzpicture}

%% file: fig/free_particle_logH.tex
\begin{tikzpicture}

\begin{axis}[
tick align=outside,
tick pos=left,
width=0.45\textwidth,
x grid style={white!69.0196078431373!black},
xlabel={\(\displaystyle t\)},
xmin=-2.45, xmax=51.45,
xtick style={color=black},
y grid style={white!69.0196078431373!black},
ylabel={\(\displaystyle \log(H)\)},
ymin=-1.9178502511308, ymax=0.89734043188779,
ytick style={color=black}
]
\addplot [semithick, black]
table {%
0 0.693147180559945
0.5 0.668470192847505
1 0.643776789497692
1.5 0.619067400879906
2 0.594342445582776
2.5 0.569602330761765
3 0.544847452475263
3.5 0.5200781960096
4 0.495294936193436
4.5 0.470498037701874
5 0.445687855350745
5.5 0.420864734381397
6 0.396029010736334
6.5 0.371181011326076
7 0.346321054287526
7.5 0.321449449234188
8 0.296566497498517
8.5 0.271672492366641
9 0.246767719305836
9.5 0.221852456184885
10 0.196926973487663
10.5 0.171991534520159
11 0.147046395611113
11.5 0.122091806306615
12 0.0971280095587175
12.5 0.0721552419084383
13 0.0471737336631208
13.5 0.0221837090686286
14 -0.00281461352372575
14.5 -0.0278210214951704
15 -0.0528353078022463
15.5 -0.0778572708242959
16 -0.10288671421532
16.5 -0.12792344676032
17 -0.152967282235756
17.5 -0.178018039274045
18 -0.203075541232097
18.5 -0.228139616063552
19 -0.253210096194728
19.5 -0.278286818404172
20 -0.30336962370565
20.5 -0.328458357234447
21 -0.353552868136946
21.5 -0.378653009463309
22 -0.403758638063233
22.5 -0.428869614484634
23 -0.453985802875165
23.5 -0.479107070886564
24 -0.504233289581574
24.5 -0.529364333343595
25 -0.554500079788748
25.5 -0.57964040968044
26 -0.604785206846272
26.5 -0.629934358097284
27 -0.655087753149374
27.5 -0.680245284546949
28 -0.705406847588577
28.5 -0.73057234025481
29 -0.755741663137855
29.5 -0.780914719373244
30 -0.806091414573362
30.5 -0.831271656762735
31 -0.856455356315126
31.5 -0.881642425892276
32 -0.906832780384386
32.5 -0.932026336852149
33 -0.957223014470333
33.5 -0.982422734472932
34 -1.0076254200997
34.5 -1.03283099654426
35 -1.05803939090349
35.5 -1.08325053212827
36 -1.10846435097555
36.5 -1.13368077996183
37 -1.15889975331756
37.5 -1.1841212069431
38 -1.20934507836569
38.5 -1.23457130669745
39 -1.25979983259475
39.5 -1.2850305982184
40 -1.31026354719509
40.5 -1.33549862457965
41 -1.36073577681849
41.5 -1.3859749517139
42 -1.41121609838929
42.5 -1.43645916725539
43 -1.46170410997725
43.5 -1.48695087944221
44 -1.51219942972862
44.5 -1.53744971607541
45 -1.5627016948525
45.5 -1.58795532353183
46 -1.61321056065932
46.5 -1.63846736582752
47 -1.66372569964888
47.5 -1.68898552372974
48 -1.7142468006451
48.5 -1.73950949391395
49 -1.76477356797516
};
\addplot [semithick, black, dashed]
table {%
0 0.769377219023309
0.5 0.742485311171929
1 0.715623561104235
1.5 0.68879127723286
2 0.661987782133283
2.5 0.635212412335267
3 0.608464518113258
3.5 0.581743463276004
4 0.555048624955642
4.5 0.528379393396511
5 0.50173517174388
5.5 0.475115375832841
6 0.448519433977542
6.5 0.421946786760954
7 0.395396886825365
7.5 0.368869198663752
8 0.342363198412151
8.5 0.315878373643281
9 0.289414223161438
9.5 0.262970256798833
10 0.236545995213539
10.5 0.210140969689063
11 0.183754721935692
11.5 0.157386803893768
12 0.13103677753886
12.5 0.104704214688995
13 0.0783886968139646
13.5 0.0520898148468623
14 0.0258071689978048
14.5 -0.000459631430073079
15 -0.0267109682222659
15.5 -0.0529472144321438
16 -0.0791687345578582
16.5 -0.105375884716782
17 -0.13156901281742
17.5 -0.15774845872889
18 -0.183914554447909
18.5 -0.210067624263157
19 -0.236207984917207
19.5 -0.26233594576581
20 -0.288451808934735
20.5 -0.314555869473956
21 -0.340648415509311
21.5 -0.366729728391639
22 -0.392800082843374
22.5 -0.418859747102535
23 -0.44490898306429
23.5 -0.470948046419914
24 -0.496977186793324
24.5 -0.522996647875082
25 -0.549006667553934
25.5 -0.575007478045955
26 -0.600999306021202
26.5 -0.626982372728026
27 -0.652956894114981
27.5 -0.67892308095035
28 -0.704881138939382
28.5 -0.730831268839248
29 -0.756773666571646
29.5 -0.782708523333223
30 -0.808636025703785
30.5 -0.834556355752238
31 -0.86046969114046
31.5 -0.886376205225014
32 -0.91227606715671
32.5 -0.938169441978204
33 -0.964056490719437
33.5 -0.989937370491177
34 -1.0158122345765
34.5 -1.04168123252042
35 -1.06754451021745
35.5 -1.09340220999749
36 -1.11925447070966
36.5 -1.14510142780451
37 -1.17094321341424
37.5 -1.19677995643136
38 -1.22261178258552
38.5 -1.2484388145186
39 -1.27426117185819
39.5 -1.30007897128964
40 -1.32589232662605
40.5 -1.35170134887708
41 -1.37750614631615
41.5 -1.40330682454594
42 -1.42910348656256
42.5 -1.45489623281841
43 -1.48068516128316
43.5 -1.50647036750388
44 -1.5322519446634
44.5 -1.55802998363744
45 -1.58380457305055
45.5 -1.60957579933058
46 -1.63534374676205
46.5 -1.66110849753817
47 -1.68687013181166
47.5 -1.71262872774442
48 -1.73838436155611
48.5 -1.76413710757141
49 -1.78988703826632
};
\end{axis}

\end{tikzpicture}

%% file: fig/harmonic_oscillator_qt.tex
\begin{tikzpicture}

\begin{axis}[
tick align=outside,
tick pos=left,
width=0.45\textwidth,
x grid style={white!69.0196078431373!black},
xlabel={\(\displaystyle t\)},
xmin=-1.225, xmax=25.725,
xtick style={color=black},
y grid style={white!69.0196078431373!black},
ylabel={\(\displaystyle q\)},
ymin=-2.06040817585875, ymax=2.25811529293511,
ytick style={color=black}
]
\addplot [semithick, black]
table {%
0 0
0.5 1
1 1.74264705882353
1.5 2.0618187716263
2 1.89394153584877
2.5 1.29019834469494
3 0.401767353229074
3.5 -0.557804689641609
4 -1.36377987120027
4.5 -1.83272740902931
5 -1.86411165454994
5.5 -1.46157946831653
6 -0.729508298512478
6.5 0.153764490818493
7 0.979227828726011
7.5 1.5565281170995
8 1.75773201223172
8.5 1.54549160714355
9 0.979457691699236
9.5 0.199994748716795
10 -0.606450988775281
10.5 -1.25182491190874
11 -1.59019928684389
11.5 -1.55062682105076
12 -1.1517509643642
12.5 -0.495234280021938
13 0.259938628746272
13.5 0.935834709880703
14 1.37738944169096
14.5 1.48786481728366
15 1.24986854211768
15.5 0.727411541985869
16 0.0489997556312232
16.5 -0.6238369734024
17 -1.13490242862551
17.5 -1.36949333022858
18 -1.28001365609141
18.5 -0.895356036068782
19 -0.312276248166033
19.5 0.328784849760078
20 0.877425293382019
20.5 1.20847737833347
21 1.25045988796013
21.5 1.00084480205538
22 0.524920859879504
22.5 -0.0610718894198597
23 -0.618224586856831
23.5 -1.01780216579409
24 -1.17090097849987
24.5 -1.04811003470716
};
\addplot [semithick, black, dashed]
table {%
0 0
0.5 1
1 1.73351007378733
1.5 2.02974726389382
2 1.82787777163645
2.5 1.1890119053957
3 0.278416807203213
3.5 -0.677016756853642
4 -1.44515803986919
4.5 -1.84489486675936
5 -1.78866687587402
5.5 -1.30132779777673
6 -0.511360313388765
6.5 0.382749645375182
7 1.16223514826238
7.5 1.64144681466281
8 1.71192512859608
8.5 1.366720107583
9 0.699565527972696
9.5 -0.120271777879106
10 -0.890785532090847
10.5 -1.42688343636073
11 -1.60472485211197
11.5 -1.39015313800131
12 -0.844740414489336
12.5 -0.10853588352258
13 0.635735651888571
13.5 1.20791027982492
14 1.47388535559105
14.5 1.37690824277054
15 0.949389313022191
15.5 0.302863680942375
16 -0.400931565466393
16.5 -0.990404857651801
17 -1.32584426802289
17.5 -1.33241272030206
18 -1.01664181667732
18.5 -0.462843497620357
19 0.189196975076034
19.5 0.779390713171392
20 1.16655596758146
20.5 1.2620890335403
21 1.05009045554063
21.5 0.589414438793482
22 -0.00240776253091266
22.5 -0.579035625050525
23 -1.00141577444462
23.5 -1.17122514851999
24 -1.05363986275509
24.5 -0.684187819661455
};
\end{axis}

\end{tikzpicture}

%% file: fig/harmonic_oscillator_zt.tex
\begin{tikzpicture}

\begin{axis}[
tick align=outside,
tick pos=left,
width=0.45\textwidth,
x grid style={white!69.0196078431373!black},
xlabel={\(\displaystyle t\)},
xmin=-1.225, xmax=25.725,
xtick style={color=black},
y grid style={white!69.0196078431373!black},
ylabel={\(\displaystyle z\)},
ymin=-1.13770334051163, ymax=1.10616739776105,
ytick style={color=black}
]
\addplot [semithick, black]
table {%
0 0
0.5 0.9375
1 0.995455098399655
1.5 0.167816787251847
2 -0.786198341774351
2.5 -1.03570921604469
3 -0.399428604231016
3.5 0.529813888186605
4 0.935384082741338
4.5 0.493307599373682
5 -0.372203811119105
5.5 -0.892130395181838
6 -0.633953037765747
6.5 0.141349050294788
7 0.738975571507662
7.5 0.651898166626748
8 -0.0104362929272084
8.5 -0.647084788526955
9 -0.708973835488477
9.5 -0.170631263818626
10 0.473663828103324
10.5 0.662505400369975
11 0.255621136330009
11.5 -0.365752702031336
12 -0.653934791377127
12.5 -0.37610728947455
13 0.200121262180427
13.5 0.562586625783103
14 0.409054667859153
14.5 -0.102072013260742
15 -0.511326982635029
15.5 -0.469934774269454
16 -0.0356197594895317
16.5 0.39732763425094
17 0.455259575151726
17.5 0.106911094577479
18 -0.326498025113166
18.5 -0.466138670941935
19 -0.205651648478059
19.5 0.210431940637712
20 0.415243546660769
20.5 0.24252106870309
21 -0.1396777068729
21.5 -0.390651373545208
22 -0.299879043916575
22.5 0.0375581927004869
23 0.318198134726648
23.5 0.302618951985995
24 0.0190913837731257
24.5 -0.27405890223958
};
\addplot [semithick, black, dashed]
table {%
0 0
0.5 0.928016041051869
1 1.00417327329411
1.5 0.202369180556242
2 -0.714860531707061
2.5 -0.918639836156657
3 -0.260316863289759
3.5 0.628401641321323
4 0.940018905778733
4.5 0.423978676013293
5 -0.41579976006666
5.5 -0.809628118346431
6 -0.425260681505013
6.5 0.355485209287288
7 0.813965706323787
7.5 0.55561231767988
8 -0.152026991123853
8.5 -0.651513930864649
9 -0.504119934431786
9.5 0.127392731414092
10 0.654070535913239
10.5 0.609723755572906
11 0.0630150972293261
11.5 -0.470200881636466
12 -0.511301185754617
12.5 -0.0463031995175501
13 0.484257296037269
13.5 0.602562964085641
14 0.223836693962327
14.5 -0.286798998182246
15 -0.463702034041465
15.5 -0.163196497805122
16 0.323283233473017
16.5 0.55176994590295
17 0.331349839727923
17.5 -0.116968168587617
18 -0.37854029571641
18.5 -0.226379907521909
19 0.184435763146053
19.5 0.474472097344502
20 0.391154085018497
20.5 0.029087598682794
21 -0.271713565569019
21.5 -0.242716975148194
22 0.0757979722049635
22.5 0.385924354338248
23 0.411854092756303
23.5 0.146143866277505
24 -0.15669686248772
24.5 -0.221222254077577
};
\end{axis}

\end{tikzpicture}

%% file: fig/harmonic_oscillator_logH.tex
\begin{tikzpicture}

\begin{axis}[
tick align=outside,
tick pos=left,
width=0.45\textwidth,
x grid style={white!69.0196078431373!black},
xlabel={\(\displaystyle t\)},
xmin=-1.2, xmax=25.2,
xtick style={color=black},
y grid style={white!69.0196078431373!black},
ylabel={\(\displaystyle \log(H)\)},
ymin=-0.497917525152205, ymax=0.864040398244621,
ytick style={color=black}
]
\addplot [semithick, black]
table {%
0 0.75377180237638
0.5 0.737254214601081
1 0.724044383495714
1.5 0.703474333090872
2 0.671435739757414
2.5 0.633756603463665
3 0.601376698976027
3.5 0.580674535810622
4 0.567470651664026
4.5 0.550985917175371
5 0.523533621199151
5.5 0.486957887489165
6 0.451126003676248
6.5 0.425467615094999
7 0.410291642685679
7.5 0.396492271525865
8 0.373961334228078
8.5 0.34029041498921
9 0.302741686032156
9.5 0.272184142206276
10 0.253359891201856
10.5 0.240399722302362
11 0.222352516063092
11.5 0.192908627856452
12 0.155636681015125
12.5 0.121055351554631
13 0.0974998922388164
13.5 0.0833721502316193
14 0.0686605731918845
14.5 0.0440923339772405
15 0.00902004799281766
15.5 -0.0280495584863238
16 -0.0566314451368834
16.5 -0.0737674126097449
17 -0.08683481918604
17.5 -0.106645020596659
18 -0.137967239062085
18.5 -0.175599268510006
19 -0.208661244423591
19.5 -0.230165480624614
20 -0.24356708949499
20.5 -0.259486431111851
21 -0.286112507810329
21.5 -0.322315290781766
22 -0.35856882807483
22.5 -0.385079372921775
23 -0.400765172371498
23.5 -0.414285712159394
24 -0.436010346815985
};
\addplot [semithick, black, dashed]
table {%
0 0.802133219908402
0.5 0.777133219908402
1 0.752133219908402
1.5 0.727133219908402
2 0.702133219908402
2.5 0.677133219908402
3 0.652133219908402
3.5 0.627133219908402
4 0.602133219908402
4.5 0.577133219908402
5 0.552133219908402
5.5 0.527133219908402
6 0.502133219908402
6.5 0.477133219908402
7 0.452133219908402
7.5 0.427133219908402
8 0.402133219908402
8.5 0.377133219908402
9 0.352133219908402
9.5 0.327133219908402
10 0.302133219908402
10.5 0.277133219908402
11 0.252133219908402
11.5 0.227133219908402
12 0.202133219908402
12.5 0.177133219908402
13 0.152133219908402
13.5 0.127133219908402
14 0.102133219908402
14.5 0.0771332199084017
15 0.0521332199084017
15.5 0.0271332199084017
16 0.00213321990840177
16.5 -0.0228667800915984
17 -0.0478667800915983
17.5 -0.0728667800915983
18 -0.0978667800915983
18.5 -0.122866780091598
19 -0.147866780091598
19.5 -0.172866780091598
20 -0.197866780091598
20.5 -0.222866780091598
21 -0.247866780091598
21.5 -0.272866780091598
22 -0.297866780091598
22.5 -0.322866780091598
23 -0.347866780091598
23.5 -0.372866780091598
24 -0.397866780091598
};
\end{axis}

\end{tikzpicture}

%% file: Exact_discrete_lagrangianin_contact_mechanics.bbl
\begin{thebibliography}{}

\bibitem[Barbero Li\~{n}\'{a}n et~al., 2019]{BCGM}
Barbero Li\~{n}\'{a}n, M., Cendra, H., Garc\'{\i}a Tora\~{n}o, E., and
  Mart\'{\i}n~de Diego, D. (2019).
\newblock Morse families and {D}irac systems.
\newblock {\em J. Geom. Mech.}, 11(4):487--510.

\bibitem[Blanes and Casas, 2016]{blanes}
Blanes, S. and Casas, F. (2016).
\newblock {\em A concise introduction to geometric numerical integration}.
\newblock Monographs and Research Notes in Mathematics. CRC Press, Boca Raton,
  FL.

\bibitem[Bloch, 2015]{Bloch}
Bloch, A. (2015).
\newblock {\em Nonholonomic Mechanics and Control}.
\newblock Springer, Interdisciplinary Applied Mathematics 24.

\bibitem[Bravetti, 2017]{Bravetti2017}
Bravetti, A. (2017).
\newblock Contact {{Hamiltonian Dynamics}}: {{The Concept}} and {{Its Use}}.
\newblock {\em Entropy}, 19(12):535.

\bibitem[Bravetti, 2018]{Bravetti2018}
Bravetti, A. (2018).
\newblock Contact geometry and thermodynamics.
\newblock {\em Int. J. Geom. Methods Mod. Phys.}, 16(supp01):1940003.

\bibitem[Bravetti et~al., 2020]{BSVZ}
Bravetti, A., Seri, M., Vermeeren, M., and Zadra, F. (2020).
\newblock Numerical integration in {C}elestial {M}echanics: a case for contact
  geometry.
\newblock {\em Celestial Mech. Dynam. Astronom.}, 132(1):Paper No. 7.

\bibitem[de~León and de~Diego, 1996]{LMdD1996}
de~León, M. and de~Diego, D.~M. (1996).
\newblock On the geometry of non‐holonomic lagrangian systems.
\newblock {\em Journal of Mathematical Physics}, 37:3389--3414.

\bibitem[de~León and Lainz~Valcázar, 2019a]{deLeon2018}
de~León, M. and Lainz~Valcázar, M. (2019a).
\newblock Contact hamiltonian systems.
\newblock {\em Journal of Mathematical Physics}, 60(10):102902.

\bibitem[de~León and Lainz~Valcázar, 2019b]{deLeon2019}
de~León, M. and Lainz~Valcázar, M. (2019b).
\newblock Singular lagrangians and precontact hamiltonian systems.
\newblock {\em International Journal of Geometric Methods in Modern Physics},
  16(10):1950158.

\bibitem[de~León and Valcázar, 2020]{deLeon2020}
de~León, M. and Valcázar, M.~L. (2020).
\newblock Infinitesimal symmetries in contact hamiltonian systems.
\newblock {\em Journal of Geometry and Physics(forthcoming)}.

\bibitem[Ferraro et~al., 2017]{FLMMV}
Ferraro, S., de~Le\'{o}n, M., Marrero, J.~C., Mart\'{\i}n~de Diego, D., and
  Vaquero, M. (2017).
\newblock On the geometry of the {H}amilton-{J}acobi equation and generating
  functions.
\newblock {\em Arch. Ration. Mech. Anal.}, 226(1):243--302.

\bibitem[Gaset et~al., 2019]{Gaset2019}
Gaset, J., Gr{\`a}cia, X., {Mu{\~n}oz-Lecanda}, M.~C., Rivas, X., and
  {Rom{\'a}n-Roy}, N. (2019).
\newblock New contributions to the {{Hamiltonian}} and {{Lagrangian}} contact
  formalisms for dissipative mechanical systems and their symmetries.
\newblock {\em arXiv}.

\bibitem[Hairer et~al., 2010]{hairer}
Hairer, E., Lubich, C., and Wanner, G. (2010).
\newblock {\em Geometric numerical integration}, volume~31 of {\em Springer
  Series in Computational Mathematics}.
\newblock Springer, Heidelberg.
\newblock Structure-preserving algorithms for ordinary differential equations,
  Reprint of the second (2006) edition.

\bibitem[Herglotz, 1930]{Herglotz1930}
Herglotz, G. (1930).
\newblock Beruhrungstransformationen.
\newblock In {\em Lectures at the {{University}} of {{Gottingen}}},
  {Gottingen}.

\bibitem[Libermann and Marle, 1987]{marle}
Libermann, P. and Marle, C.-M. (1987).
\newblock {\em Symplectic geometry and analytical mechanics}, volume~35 of {\em
  Mathematics and its Applications}.
\newblock D. Reidel Publishing Co., Dordrecht.
\newblock Translated from the French by Bertram Eugene Schwarzbach.

\bibitem[Marrero et~al., 2016]{MMdDM2016}
Marrero, J., de~Diego, D.~M., and Martínez, E. (2016).
\newblock On the exact discrete lagrangian function for variational
  integrators: theory and applications. arxiv:1608.01586v1 [math.dg].

\bibitem[Marsden and West, 2001]{marsden-west}
Marsden, J.~E. and West, M. (2001).
\newblock Discrete mechanics and variational integrators.
\newblock {\em Acta Numer.}, 10:357--514.

\bibitem[Patrick and Cuell, 2009]{PatrickCuell}
Patrick, G.~W. and Cuell, C. (2009).
\newblock Error analysis of variational integrators of unconstrained
  {L}agrangian systems.
\newblock {\em Numer. Math.}, 113(2):243--264.

\bibitem[Sanz-Serna and Calvo, 1994]{serna}
Sanz-Serna, J.~M. and Calvo, M.~P. (1994).
\newblock {\em Numerical {H}amiltonian problems}, volume~7 of {\em Applied
  Mathematics and Mathematical Computation}.
\newblock Chapman \& Hall, London.

\bibitem[Simoes~A. and de~Diego~D., 2020]{AMM}
Simoes~A., M.~J. and de~Diego~D., M. (2020).
\newblock Exact discrete lagrangian mechanics for nonholonomic mechanics.

\bibitem[Vermeeren et~al., 2019]{VBS}
Vermeeren, M., Bravetti, A., and Seri, M. (2019).
\newblock Contact variational integrators.
\newblock {\em J. Phys. A}, 52(44):445206, 28.

\end{thebibliography}
